\documentclass[12pt,draftclsnofoot,journal,onecolumn]{IEEEtran}

\makeatletter
\def\ps@headings{%
\def\@oddhead{\mbox{}\scriptsize\rightmark \hfil \thepage}%
\def\@evenhead{\scriptsize\thepage \hfil \leftmark\mbox{}}%
\def\@oddfoot{}%
\def\@evenfoot{}}
\makeatother
\usepackage[dvips]{color}
\usepackage{epsf}
\usepackage{times}
\usepackage{epsfig}
\usepackage{graphicx}
\usepackage{amsmath}
\usepackage{amssymb}
\usepackage{amsxtra}
\usepackage{here}
\usepackage{rawfonts}
\usepackage{times}
\usepackage{url}
\usepackage{cite}
\usepackage{algorithm}
\usepackage{algorithmic}
\newtheorem{theorem}{\bf Theorem}

\newtheorem{proposition}{\bf Proposition}

\newtheorem{definition}{\bf Definition}
\newtheorem{remark}{\bf Remark}
\newlength{\aligntop}
\newlength{\alignbot}
\makeatletter
\renewenvironment{align}{%
  \vspace{\aligntop}
  \start@align\@ne\st@rredfalse\m@ne
}{%
  \math@cr \black@\totwidth@
  \egroup
  \ifingather@
    \restorealignstate@
    \egroup
    \nonumber
    \ifnum0=`{\fi\iffalse}\fi
  \else
    $$%
  \fi
  \ignorespacesafterend%
  \vspace{\alignbot}\par\noindent
}

\IEEEoverridecommandlockouts
\begin{document}

\title{{\huge Coalitional Games in Partition Form for Joint Spectrum Sensing and Access in Cognitive Radio Networks} }

\author{Walid Saad, Zhu Han, Rong Zheng, Are Hj{\o}rungnes, Tamer Ba\c{s}ar, and H. Vincent Poor \vspace{-1.3cm} \thanks{W.~Saad is
with the Electrical and Computer Engineering Department, University of Miami, Coral Gables, FL, USA, email: \texttt{walid@miami.edu}.
Z.~Han is with the Electrical and Computer Engineering Department, University of Houston, Houston, USA, email: \texttt{zhan2@mail.uh.edu}. R. Zheng is with the Computer Science Department, University of Houston, Houston, USA, email: \texttt{rzheng@cs.uh.edu}. A. Hj{\o}rungnes was with the UNIK - Graduate University Center, University of Oslo, Oslo, Norway. T. Ba\c{s}ar is with the Coordinated Science Laboratory, University of Illinois at Urbana Champaign, IL, USA, email: \texttt{basar1@illinois.edu}.  H. V. Poor is
with the Electrical Engineering Department, Princeton University, Princeton, NJ, USA, email: \texttt{poor@princeton.edu}. This  research is supported by NSF Grants CNS-0910461, CNS-0953377, CNS-0905556, ECCS-1028782, CNS-1117560, and CNS-0832084, by the Research Council of Norway through the projects 183311/S10, 176773/S10 and 18778/V11, by AFOSR Grant FA9550-09-1-0249, and by DTRA under Grant HDTRA1-07-1-0037.}}

\date{}
\maketitle

\thispagestyle{empty}
\begin{abstract}
Unlicensed secondary users~(SUs) in cognitive radio networks are subject to an inherent tradeoff between \emph{spectrum sensing} and \emph{spectrum access}. Although each SU has an incentive to sense the primary user~(PU) channels for locating spectrum holes, this exploration of the spectrum can come at the expense of a shorter transmission time, and, hence, a possibly smaller capacity for data transmission. This paper investigates the impact of this tradeoff on the cooperative strategies of a network of SUs that seek to cooperate in order to improve their view of the spectrum (sensing), reduce the possibility of interference among each other, and improve their transmission capacity (access). The problem is modeled as a coalitional game in \emph{partition form} and an algorithm for coalition formation is proposed. Using the proposed algorithm, the SUs can make individual distributed decisions to join or leave a coalition while maximizing their utilities which capture the average time spent for sensing as well as the capacity achieved while accessing the spectrum. It is shown that, by using the proposed algorithm, the SUs can self-organize into a network partition composed of disjoint coalitions, with the members of each coalition cooperating to jointly optimize their sensing and access performance. Simulation results show the performance improvement that the proposed algorithm yields with respect to the non-cooperative case. The results also show how the algorithm allows the SUs to self-adapt to changes in the environment such as the change in the traffic of the PUs, or slow mobility.
\end{abstract}
\indent \indent \indent \indent %\small
\vspace{-0.2cm}
\newpage
 \setcounter{page}{1}
\section{Introduction}
With the ongoing growth in wireless services, the demand for the radio spectrum has significantly increased. However, the radio spectrum is limited and much of it has already been licensed to existing operators. Numerous studies conducted by agencies such as the Federal Communications Commission~(FCC) in the United States have shown that the actual licensed spectrum remains unoccupied for large periods of time \cite{FCC}. Thus, \emph{cognitive radio}~(CR) systems have been proposed \cite{CR01} in order to efficiently exploit this under-utilized spectrum. Cognitive radios or secondary users~(SUs) are wireless devices that can intelligently monitor and adapt to their environment and, hence, they are able to share the spectrum with the licensed primary users~(PUs), operating whenever the PUs are idle. Implementing cognitive radio systems faces various challenges \cite{CR02,MAT1}, notably, for spectrum sensing and spectrum access. Spectrum sensing mainly deals with the stage during which the SUs attempt to learn their environment prior to the spectrum access stage when the SUs actually transmit their data.

\subsection{Existing Work on Spectrum Sensing and Access}
Existing work has considered various aspects of spectrum sensing and spectrum access, individually. In \cite{CS00}, the performance of spectrum sensing, in terms of throughput, is investigated when the SUs share their instantaneous knowledge of the channel. The work in \cite{DT00} studies the performance of different detectors for spectrum sensing, while in \cite{YY02}, the sensing time that maximizes the achievable throughput of the SUs, given a detection-false alarm rate is derived. The authors in \cite{GL00} study the use of a novel approach for collaborative spectrum sensing based on abnormality detection. Different cooperative techniques for improving spectrum sensing performance are discussed in \cite{YY01,GL02,GL03,CS01,CS04,WS01,TVT}. Further, spectrum access has also received considerable attention \cite{SA00,KCC00,SA01,SA02,SA03,SA04,SA06,KCC01}. In \cite{SA00}, a dynamic programming approach is proposed to allow the SUs to maximize their channel access times while taking into account a penalty factor from any collision with the PU. The work in \cite{SA00} (and the references therein) establish that, in practice, the sensing time of CR networks is \emph{large} and affects the access performance of the SUs. The authors in \cite{KCC00} propose a novel multiple access scheme that takes into account the physical layer transmission in cognitive networks. In \cite{SA01}, the authors study the problem of interference alignment in a cognitive radio network consisting of SUs having multiple antennas. Non-cooperative solutions for dynamic spectrum access are proposed in \cite{SA02} while taking into account changes in the SUs' environment such as the arrival of new PUs, among others. Additional challenges of spectrum access are studied in \cite{SA03,SA04,SA06,KCC01}.

Clearly, the spectrum sensing and spectrum access aspects of cognitive networks have been widely investigated, independently. However, a key challenge that remains relatively unexplored is to study the tradeoff between \emph{spectrum sensing} and \emph{spectrum access} when the SUs seek to improve both aspects, \emph{jointly}. This tradeoff arises from the fact that the sensing time for the SUs is non-negligible \cite{SA00}, and can reduce their transmission performance. Thus, although each SU has an incentive to sense as many PU channels as possible for locating access opportunities, this spectrum exploration may come at the expense of a smaller transmission time, and, hence, a possibly smaller capacity for data transmission. Also, due to the limited capability of the cognitive devices, each SU, on its own, may not be able to explore more than a limited number of channels. As a result, the SUs can rely on cooperation for sharing the spectrum knowledge with nearby cognitive radios. Therefore, it is important to design cooperative strategies that allow the SUs to improve their performance while taking into account \emph{both} sensing and access metrics.

\subsection{Our Contributions}
The main contribution of this paper is to devise a cooperative scheme among the SUs in a multi-channel cognitive network, which enables them to improve their performance jointly during spectrum sensing and access. From a sensing perspective, we propose a scheme through which the SUs cooperate in order to share their channel knowledge so as to improve their view of the spectrum and reduce their sensing times. From an access perspective, the proposed cooperation protocol allows the SUs to improve their access capacities by: (i)- learning from their cooperating partners the existence of alternative channels with better conditions, (ii)- reducing the mutual interference, and (iii)- exploiting multiple channels simultaneously, when possible. We model the problem as a coalitional game in \emph{partition form}, and we propose an algorithm for coalition formation. Although coalitional games in partition form have been widely studied in game theory, to the best of our knowledge, no existing work has utilized the partition form of coalitional game theory in the design of cognitive radio systems. The proposed coalition formation algorithm allows the SUs to make distributed decisions to join or leave a coalition, while maximizing their utility which accounts for the average time needed to locate an unoccupied channel (spectrum sensing) and the average capacity achieved when transmitting the data (spectrum access).  Thus, the SUs self-organize into disjoint coalitions that constitute a Nash-stable network partition. Within every formed coalition, the SUs act cooperatively by sharing their views of the spectrum, coordinating their sensing orders, and distributing their powers over the seized channels whenever possible. Also, the proposed coalition formation algorithm allows the SUs to adapt the topology to environmental changes such as the changes in the availability of the PU channels or the slow mobility of the SUs. Simulation results assess the performance of the proposed algorithm relative to the non-cooperative case.

Note that, in our previous work \cite{WS01,TVT}, we have studied the use of cooperative games for performing collaborative spectrum sensing in order to improve the tradeoff between the probability of detection of the PU and the probability of false alarm for a single-channel cognitive network. However, our work in \cite{WS01,TVT} is focused on detection techniques and does not take into account spectrum access, interference coordination, spectrum view sharing, capacity optimization, sensing time, or other the key factors in spectrum sharing and access. To this end, the solutions, methods, and models studied in \cite{WS01,TVT} are inapplicable for the problem we address here, both from an application/systme perspective as well as from the game theoretic perspective (due to the need for games in partition form which is a class of cooperative games that is significantly different from the characteristic form adopted in \cite{WS01,TVT}).

The rest of this paper is organized as follows: Section~\ref{sec:noncoop} presents
the non-cooperative spectrum sensing and access model. In Section \ref{sec:jssa}, we present the proposed cooperation model for joint spectrum access and sensing, while in Section~\ref{sec:game} we model the problem using coalitional games in partition form and we devise a distributed algorithm for coalition formation. Simulation results are presented and analyzed in Section \ref{sec:sim}. Finally, conclusions are drawn in
Section \ref{sec:conc}.
\section{Non-cooperative Spectrum Sensing and Access}\label{sec:noncoop}
In this section, we present the non-cooperative procedure for spectrum sensing and access in a cognitive network, prior to proposing, in the next sections, cooperation strategies for improving the performance of the SUs jointly for sensing and access. A summary of the notation used throughout this paper is shown in Table~\ref{table:par}.
\begin{table}[e]
\scriptsize
\caption{Summary of the General Notation}
\centering
\begin{tabular}{|c | c|}
  \hline
  \textbf{Notation} &  \textbf{Explanation}  \\
  \hline\hline
 $\mathcal{N}$ & Set of $N$ secondary users \\ \hline
$\mathcal{K}$ & Set of $K$ primary users or channels \\ \hline
$S$ & A coalition of SUs, i.e., $S \subseteq \mathcal{N}$ \\ \hline
$ \mathfrak{P}$ & Set of all partitions of $\mathcal{N}$ \\ \hline
$ \Pi$ & A partition of $\mathcal{N}$ \\ \hline
$\mathcal{K}_i \subseteq \mathcal{K}$ & Subset of channels known by a certain SU $i \in \mathcal{N}$\\ \hline
 $\theta_k$   & Probability that channel $k$ is available  \\ \hline
$g_{i,k}$ & Channel gain experienced by an SU $i$ over a channel $k$ \\ \hline
$w_{i,k}$ & Weight assigned by an SU $i$ for a channel $k$ \\ \hline
  $\mathcal{K}^{\textrm{ord}}_i$ &  The set of channels used by an SU $i \in \mathcal{N}$ \emph{ordered} (non-cooperatively) \\ \hline
   $k_j$&A channel that is the $j$th element of the ordered set  $\mathcal{K}^{\textrm{ord}}_i$\\\hline
$\mathcal{K}_S = \cup_{i\in S} \mathcal{K}_i$ & The set of all channels known by a coalition $S$\\ \hline
$\mathcal{K}_i^S$ & The \emph{ordered} set of all channels used by an SU $i$ inside coalition $S$. This set results from the cooperative sorting in Algorithm~\ref{alg:sort}. \\ \hline
$\mathcal{B}_S=\{b_1,\ldots,b_{|S|}\}$  & Tuple with every element $b_i$ representing a channel in $\mathcal{K}_i^S$\\ \hline
$\mathfrak{F}_S$ & Family (or collection) of all channel tuples $\mathcal{B}_S$ for a coalition $S$\\ \hline
$\boldsymbol{P}_i^{\mathcal{B}_S} $ & $1\times |\mathcal{K}_{S_l}|$ vector  where each element $P_{i,k}^{\mathcal{B}_S }$ represents the power that SU $i \in S_l$ will use on channel $k \in \mathcal{K}_{S_l}$ given the selection $\mathcal{B}_S $\\ \hline
  \end{tabular}\label{table:par}
\end{table}

\subsection{Network Model}
Consider a cognitive radio network with $N$ secondary users~(SUs) engaged in the sensing of $K$ primary users'~(PUs) channels in order to access the spectrum and transmit their data to a common base station~(BS). Let $\mathcal{N}$ and $\mathcal{K}$ denote the set of SUs and the set of PUs (channels), respectively. Due to the random nature of the traffic of the PUs and to the dynamics of the PUs, each channel $k \in \mathcal{K}$ is available for use by the SUs with a probability of $\theta_k$ (which depends on PU traffic only and not on the SUs). Although for very small $K$ the SUs may be able to learn the statistics (probabilities $\theta_k$) of all $K$ channels,  we consider the generalized case where each SU $i \in \mathcal{N}$ can only have accurate statistics regarding a subset  $\mathcal{K}_i \subseteq \mathcal{K}$ of $K_i \le K$ channels (e.g., via standard learning algorithms), during the period of time the channels remain stationary. We consider a frequency selective channel, whereby the channel gain $g_{i,k}$ of any SU $i\in\mathcal{N}$ experienced at the BS when SU $i$ transmits over channel $k\in \mathcal{K}_i$ is %\begin{align}\label{eq:chan}
$g_{i,k}=a_{i,k}\cdot d_{i}^{-\mu}$,
%\end{align}
 with $d_{i}$ the distance between SU $i$ and the BS, $\mu$ the path loss exponent, and $a_{i,k}$ a Rayleigh distributed fading amplitude for SU $i$ on channel $k$ with a variance of $1$. We consider a channel with \emph{slow} fading which varies independently over the frequencies (quasi-static over the frequency band). Note that other channel types can also be accommodated.

\subsection{Non-cooperative Sensing Process}
For transmitting its data, each SU $i\in\mathcal{N}$ is required to sense the channels in $\mathcal{K}_i$ persistently, one at a time, in order to locate a transmission opportunity. We consider that each SU $ i \in \mathcal{N}$ is opportunistic which implies that SU $i$ senses the channels in $\mathcal{K}_i$ in a certain order, sequentially, and once it locates a spectrum hole it ends the sensing process and transmits over the first channel found unoccupied (by a PU). For the purpose of finding a preferred order for sensing, each SU $i$ assigns a weight $w_{i,k}$ to every channel $k\in\mathcal{K}_i$ which will be used in sorting the channels. When assigning the weights and ordering of the channels, the SUs face a tradeoff between improving their sensing times by giving a higher weight to channels that are often available, and improving their access performance by giving a higher weight to channels with better conditions.  The weights can be a function of a variety of parameters such as channel, interference, data, or others. Hereinafter, without loss of generality and in order to capture the joint sensing and access tradeoff, the weight $w_{i,k}$ assigned by an SU $i$ to a channel $k \in \mathcal{K}_i$ will be taken as
\begin{align}\label{eq:w}
w_{i,k}= \theta_k \cdot g_{i,k},
\end{align}
where $g_{i,k}$ is the channel gain experienced by SU $i$ over channel $k$ and $\theta_k$ is the probability that channel $k$ is available. Clearly, the weight given in (\ref{eq:w}) provides a balance between the need for quickly finding an available channel and the need for good channel conditions. Given the channel weights, each SU $i \in \mathcal{N}$ sorts its channels in decreasing order of weights and begins sensing these channels in an ordered manner. Hence, each SU $i$ senses the channels consecutively starting with the channel having the highest weight until finding an unoccupied channel on which to transmit, if any. The set of channels used by an SU $i \in \mathcal{N}$ \emph{ordered} non-cooperatively by decreasing weights is denoted by $\mathcal{K}^{\textrm{ord}}_i=\{k_1,\ldots,k_{K_i}\}$ where $w_{i,k_1} \ge w_{i,k_2} \ge \ldots \ge w_{i,k_{K_i}}$. Note that, other weights can also be adopted with little changes to the analysis in the remainder of this paper.

We consider a time-slotted spectrum sensing and access process whereby, within each slot, each SU $i \in \mathcal{N}$ spends a certain fraction of the slot for sensing the channels, and, once an available channel is found, the remaining time of the slot is used for spectrum access. In this regard, we consider that the channel available/busy time is comparable or larger to the duration of a slot, which is a common assumption in the literature \cite{CS00,CS02,CS04,SA00}.  Given the ordered set of channels $\mathcal{K}^{\textrm{ord}}_i$, the average fraction of time $\tau_i$ spent by any SU $i\in \mathcal{N}$ for locating a free channel, i.e., the average sensing time, is given by (the duration of a slot is normalized to $1$)
\begin{align}\label{eq:sensetime}
\tau_i(\mathcal{K}^{\textrm{ord}}_i) = \sum_{j=1}^{K_i} \left( j \cdot \alpha \cdot \theta_{k_j}\prod_{m=1}^{j-1} (1-\theta_{k_m})\right)+  \prod_{l=1}^{K_i} (1-\theta_{k_l})
\end{align}
where $\alpha < 1$ is the fraction of time needed for sensing a single channel, and $\theta_{k_j}$ is the probability that channel $k_j \in \mathcal{K}_i^{\textrm{ord}}$ is unoccupied.
 The first term in (\ref{eq:sensetime}) represents the average time spent for locating an unoccupied channel among the known channels in $\mathcal{K}_i^{\textrm{ord}}$, and the second term represents the probability that no available channel is found (in this case, the SU remains idle in the slot). Note that $\tau_i(\mathcal{K}^{\textrm{ord}}_i)$ is function of  $\mathcal{K}_i^{\textrm{ord}}$ and, hence, depends on the assigned weights and the ordering. For notational convenience, the argument of $\tau_i$ is dropped hereafter since the dependence on the channel ordering is clear from the context.

\subsection{Non-cooperative Utility Function}
When the SUs are acting in a non-cooperative manner, given the ordered set of channels $\mathcal{K}^{\textrm{ord}}_i$, the average capacity achieved by an SU $i \in \mathcal{N}$ is given by
\begin{align}\label{eq:avc}
C_i = \sum_{j=1}^{K_i} \theta_{k_j}\prod_{m=1}^{j-1} (1-\theta_{k_m}) \cdot \operatorname{\mathbb{E}}_{I_{i,k_j}}\left[C_{i,k_j}\right]
\end{align}
where $\theta_{k_j}\prod_{m=1}^{j-1} (1-\theta_{k_m})$ is the probability that SU $i$ accesses channel $k_j \in \mathcal{K}_i^{\textrm{ord}}$ given the ordered set $\mathcal{K}_i^{\textrm{ord}}$, and $\operatorname{\mathbb{E}}_{I_{i,k_j}}\left[C_{i,k_j}\right]$ is the expected value of the  capacity achieved by SU $i$ over channel $k_j$ with the expectation taken over the distribution of the total interference $I_{i,k_j}$ experienced on channel $k_j$ by SU $i$ from the SUs in $\mathcal{N} \setminus \{i\}$.

For evaluating the capacity in (\ref{eq:avc}), every SU $i\in \mathcal{N}$ must have perfect knowledge of the channels that the other SUs are using, as well as the order in which these channels are being sensed and accessed (to compute the expectation) which is quite difficult in a practical network. To alleviate the information needed for finding the average capacity, some works such as \cite{MM00} and \cite{LA00} consider, in (\ref{eq:avc}), the capacities under the worst case interference, instead of the expectation over the interference. However, applying this assumption in our case requires considering the capacities under worst case interference on \emph{every} channel for every SU $i$ which is quite restrictive. Thus, in our setting, as an alternative to the expectation in (\ref{eq:avc}), for any SU $i \in \mathcal{N}$ we consider the capacity $\bar{C}_{i,k_j}$ achieved over channel $k_j \in \mathcal{K}^{\textrm{ord}}_i$ under the average interference resulting from the SUs in $\mathcal{N} \setminus \{i\}$, given by
\begin{align}\label{eq:capint}
\bar{C}_{i,k_j}= \log_2{(1+\Gamma_{i,k_j})}.
\end{align}
 Here, $\Gamma_{i,k_j}$ is the SINR achieved by SU $i$ when using channel $k_j$ given an average total interference $\bar{I}_{i,k_j}$ arising from the SUs in $\mathcal{N} \setminus \{i\}$ and is given by
\begin{align}
\Gamma_{i,k_j}= \frac{g_{i,k_j} \cdot  P_{i,k_j}}{\sigma^2 + \bar{I}_{i,k_j}},
\end{align}
where $P_{i,k_j}$ is the maximum transmit power of SU $i$ used on channel $k_j$, and $\sigma^2$ is the variance of the Gaussian noise. In the non-cooperative setting, $P_{i,k_j}=\tilde{P}$ where $\tilde{P}$ is the maximum transmit power of any SU ( $\tilde{P}$ is assumed to be the same for all SUs with no loss of generality). In a practical cognitive network, through measurements, any SU $i\in\mathcal{N}$ can obtain from its receiver an estimate of the average total interference $\bar{I}_{i,k_j}$ experienced on any channel $k_j \in \mathcal{K}^{\textrm{ord}}_i$ \cite{PROAKIS}, and, thus, SU $i$ is able to evaluate the capacity in (\ref{eq:capint}). By using (\ref{eq:capint}), we define the average capacity $\bar{C}_i$ in a manner analogous to (\ref{eq:avc}) as follows:
\begin{align}\label{eq:avc2}
\bar{C}_i = \sum_{j=1}^{K_i} \theta_{k_j}\prod_{m=1}^{j-1} (1-\theta_{k_m}) \cdot\bar{C}_{i,k_j}.
\end{align}
Clearly, given the measurement of the external interference, every SU $i$ can easily evaluate its capacity in (\ref{eq:avc2}). Due to properties such as Jensen's inequality,  (\ref{eq:avc2}) represents a lower bound of (\ref{eq:avc}) but it provides a good indicator of the access performance of the SUs. Hereafter, we solely deal with capacities given the measured average interference.

Consequently, the non-cooperative utility achieved by any SU $i\in\mathcal{N}$ per slot is given by
\begin{align}\label{eq:utilnc}
u(\{i\},\mathcal{N}) = \bar{C}_i \cdot (1-\tau_i),
\end{align}
where the dependence on $\mathcal{N}$ indicates the dependence of the utility on the external interference when the SUs are non-cooperative, $\tau_i$ is the fraction of time used for sensing given by (\ref{eq:sensetime}), and $\bar{C}_i$ the average capacity given by (\ref{eq:avc2}). This utility captures the tradeoff between exploring the spectrum, i.e., sensing time, and exploiting the best spectrum opportunities, i.e., capacity achieved during spectrum access.

\section{Joint Spectrum Sensing and Access Through Cooperation}\label{sec:jssa}
In this section, we propose a cooperative scheme that enables the SUs to share their knowledge of the radio spectrum and improve their spectrum sensing and access performance, jointly.

\subsection{Cooperative Sharing of Channel Knowledge}
To improve their joint sensing and access performance, the SUs in the cognitive network can cooperate. Hence, any group
\begin{algorithm}[e]
\caption{\footnotesize Proposed sorting algorithm for any coalition $S \subseteq \mathcal{N}$}
\label{alg:sort}
\begin{algorithmic}
\footnotesize
\STATE $\mathcal{Q}_{i,0}\leftarrow\emptyset$
\FOR[For rank $r=1$ we find all the channels that SUs in $S$ sense first, for $r=2$ the channels that they sense second, and so on.]{$r=1$ to $K_S$}
\STATE $\mathcal{Q}_{i,r} \leftarrow \mathcal{Q}_{i,r-1}$, $\mathcal{K}_{r,S}\leftarrow\emptyset$, $\mathcal{R}_{r}\leftarrow\emptyset$
\STATE For rank $r$, each SU $i \in S$ proposes to select the channel $l_i^{r}$ in $\mathcal{K}_S \setminus \mathcal{Q}_{i,r}$ which has the highest weight, i.e., $l_i^{r}=\underset{k\in\mathcal{K}_S \setminus\mathcal{Q}_{i,r}}{\operatorname{arg\,max}}w_{i,k}$.
\FORALL {$i\in S \textrm{ s. t. } l_i^{r} \neq l_j^{r},\ \forall j \in S,\ i \neq j$}
\STATE SU $i$ fixes its selection for this rank, and, hence:
\STATE $\mathcal{Q}_{i,r} \leftarrow \mathcal{Q}_{i,r} \cup l_i^{r}$, $\mathcal{K}_{r,S} \leftarrow \mathcal{K}_{r,S} \cup l_i^{r}$, $\mathcal{R}_{r} \leftarrow \mathcal{R}_{r} \cup \{i\}$.
\ENDFOR
\FORALL {$G \subseteq S \setminus \mathcal{R}_{r} , \textrm{ s. t. } l_i^{r} = l_j^{r}=l_G^{r},\ \forall i,j \in G$}
\STATE a) The SU $j \in G$ which has the highest weight for $l_G^{r}$, i.e.,  $j =\underset{j\in G}{\operatorname{arg\,max}}\ w_{j,k_G^{r}}$, selects channel $l_G^{r}$ for rank $r$.
\STATE b)  $\mathcal{Q}_{j,r} \leftarrow \mathcal{Q}_{j,r} \cup l_G^{r}$, $\mathcal{K}_{r,S} \leftarrow \mathcal{K}_{r,S} \cup l_G^{r}$, $R_{r} \leftarrow R_{r} \cup \{j\}$.
\IF[SUs with unselected channels for $r$ exist]{$\mathcal{R}_r \neq S$}
\STATE The SUs in $S \setminus \mathcal{R}_r $ repeat the previous procedure, but each SU $i \in S \setminus \mathcal{R}_r $, can only use the channels in $\mathcal{K_S} \setminus\mathcal{K}_{r,S} \cup \mathcal{Q}_{i,r}$. However, if for any SU $i \in S \setminus \mathcal{R}_r $, we have $\mathcal{K_S} \setminus \mathcal{K}_{r,S} \cup \mathcal{Q}_{i,r} = \emptyset$, then this SU will simply select the channel that will maximize its weight from the set $\mathcal{K_S} \setminus \mathcal{Q}_{i,r}$, regardless of the other SUs selection.
\ENDIF
\ENDFOR
\ENDFOR
\end{algorithmic}
\end{algorithm}
\noindent of SUs can cooperate by forming a \emph{coalition} $S \subseteq \mathcal{N}$ in order to: (i)- improve their sensing times and learn the presence of channels with better conditions by exchanging information on the statistics of their known channels, (ii)- jointly coordinate the order in which the channels are accessed to reduce the mutual interference, and (iii)- share their instantaneous sensing results to improve their capacities by distributing their total power over multiple channels, when possible.

First and foremost, whenever a coalition $S$ of SUs forms, its members exchange their knowledge on the channels and their statistics. Hence, the set of channels that the coalition is aware of can be given by $\mathcal{K}_S = \cup_{i\in S} \mathcal{K}_i$ with cardinality $|\mathcal{K}_S|=K_S$. By sharing this information, each member of $S$ can explore a larger number of channels, and, thus, can improve its sensing time by learning channels with better availability and by reducing the second term in (\ref{eq:sensetime}).  Moreover, as a result of sharing the known channels, some members of $S$ may be able to access the spectrum with better channel conditions, thereby, possibly improving their capacities as well.

\subsection{Proposed Algorithm for Cooperative Interference Management and Channel Sorting}
Once the coalition members share their knowledge about the channels, the SUs will jointly coordinate their order of access over the channels in $\mathcal{K}_S$ in order to minimize the probability of interfering with each other. In this context, analogously to the non-cooperative case, the SUs in $S$ proceed by assigning different weights on the channels in $\mathcal{K}_S$ using (\ref{eq:w}). Then, the SUs in coalition $S$ \emph{cooperatively} sort their channels, in a manner to reduce interference as much as possible. Thus, the SUs jointly \emph{rank} their channels on a rank scale from $1$ (the first channel to sense) to $K_S$ (the last channel to sense). For every SU $i\in S$, let $\mathcal{Q}_{i,r}$ denote the set of channels that SU $i$ has selected \emph{until} and including rank $r$. Further, we denote by $\mathcal{R}_r$ the set of SUs that have selected a channel for rank $r$ and by $\mathcal{K}_{r,S}$ the set of channels that have been  selected for rank $r$ by members of $S$. Given this notation (summarized in Table~\ref{tab:notalg1}), we propose the sorting procedure in Algorithm~\ref{alg:sort} for any coalition $S$.
\begin{table}[e]
\scriptsize
\caption{Summary of the Notation Specific for Algorithm~\ref{alg:sort}}
\centering
\begin{tabular}{|c | c|}
  \hline
  \textbf{Notation} &  \textbf{Explanation}  \\
  \hline\hline
$\mathcal{Q}_{i,r}$  &  Set of channels that an SU $i$ has selected in Algorithm~\ref{alg:sort} \emph{until} and including a channel rank $r$  \\ \hline
$\mathcal{R}_r$ &  Set of SUs that have selected a channel for rank $r$ in Algorithm~\ref{alg:sort}  \\ \hline
$\mathcal{K}_{r,S}$ &  Set of channels that have been  selected for rank $r$ by members of $S$  \\ \hline
$c_i^{r}$&  Channel selected by an SU $i$ (member of coalition $S$) at rank $r$ during  Algorithm~\ref{alg:sort}    \\ \hline
  \end{tabular}\label{tab:notalg1}
\end{table}

Essentially, in order to apply Algorithm~\ref{alg:sort}, the members of a cooperative coalition $S$ proceed as follows. First, every SU in coalition $S$ starts by applying the non-cooperative weighting procedure over the set of channels $\mathcal{K}_S$. Initially, the SU already performs this weighting procedure when acting non-cooperatively. Subsequently, the SUs share their current ordering of the channels over a signalling channel, e.g., a temporary ad hoc channel which is commonly used in ad hoc cognitive radio~\cite{ND01}. The SUs cooperatively inspect the received rankings of channels while proceeding sequentially by rank (i.e., they check the top ranked channel of all SUs first, then move to the next rank, and so on). At a given rank $r$, all SUs that have chosen a certain channel that does not conflict with the choices of the other SUs will actually be offered this channel by the coalition. In contrast, if, at a given rank $r$, the cooperating SUs find out that the same channel has been selected by a set of SUs $G \subseteq S$, then,  SU $j \in G$ with the highest weight is assigned this channel at rank $r$. Subsequently this SU $j$ will no longer participate in bidding for channels at rank $r$. As long as there exist SUs that have not made their channel selection at rank $r$, i.e., $\mathcal{R}_r \neq S$, then these SUs repeat the same procedure as above (i.e., re-rank their channels and manage conflicting channel selections by offering the channel to the SU that values it the most, i.e., with the highest weight) but can use only the channels that their partners have not already selected at rank $r$. However, whenever an SU $i \in S \setminus \mathcal{R}_r $ can no longer choose a channel not used by the others at rank $r$, it is inevitable that this SU $i$ interferes with some of its partners at rank $r$, then  SU $i$ simply selects, at rank $r$, the channel in $\mathcal{K_S} \setminus \mathcal{Q}_{i,r}$ with the highest weight. Hence, in summary, the SUs inside a certain coalition $S$ can use Algorithm~\ref{alg:sort} to share their valuation or ranking of the channels and, subsequently, coordinate their order of access over these channels so as to avoid interference. As a result of the sorting process, each SU $i \in S$ will have an \emph{ordered set of channels}  $\mathcal{K}_i^S$ of cardinality $K_S$ which reflects the result of Algorithm~\ref{alg:sort}.

Note that, in order to implement Algorithm~\ref{alg:sort}, each SU in coalition $S$ needs to share its own ranking of the channels in $\mathcal{K}_S$. Essentially, this ranking is the only information that needs to be exchanged so that the SUs can execute Algorithm~\ref{alg:sort}. Following this exchange, the cooperating SUs would combine these channel rankings to modify and update the overall ranking so as to reduce interference based on  Algorithm~\ref{alg:sort}. We do note that, throughout this process, the SUs would need to regularly exchange this ranking information. In practice, this information exchange can be done over a signalling channel such as the temporary ad hoc channel and it will not require a large overhead since the SUs need to only share the ``ranking'' of the channels and not the way in which they actually rank the channels (i.e., they do not need to reveal their ranking method nor exchange their channel gains or locations). This overhead is also reduced by the fact that, at a given rank $r$, the SUs that have already obtained their channel assignment do not need to further share the information exchange with the remaining SUs.

Given this new ordering resulting from the sorting procedure of Algorithm~\ref{alg:sort}, for every SU $i \in S$, the total average sensing time $\tau_i^S$ will still be expressed by (\ref{eq:sensetime}). However, the sensing time  $\tau_i^S$ is a function of the channel ordering based on the set $\mathcal{K}_i^S$ which is ordered cooperatively, rather than $\mathcal{K}^{\textrm{ord}}_i$ which is the non-cooperative ordering.

Using Algorithm~\ref{alg:sort}, the SUs that are members of the same coalition are able to reduce the interference on each other, by minimizing the possibility of selecting the same channel at the same rank (although they can still select the same channel but at different ranks). However, as a result of this joint sorting, some SUs might need to give a high rank to some channels with lower weights which can increase the sensing time of these SUs. Hence, this cooperative sorting of the channels highlights the fact that some SUs may trade off some gains in sensing performance (obtained by sharing channel statistics) for obtaining access gains (by avoiding interference through joint sorting). As we will see later in this section, in addition to the interference reduction, some SUs in a coalition $S$ can also obtain access gains by using multiple channels simultaneously.

\subsection{Cooperative Power Allocation and Coalitional Utility}
For every coalition $S$, we define $\mathcal{B}_S=\{b_1,\ldots,b_{|S|}\}$ as the tuple with every element $b_i$ representing a channel in $\mathcal{K}_i^S$ selected by SU $i \in S$. Denote by $\mathfrak{F}_S$ the family of all such tuples for coalition $S$ which corresponds to the family of all permutations, with repetition, for the SUs in $S$ over the channels in $\mathcal{K}_S$. Each tuple  $\mathcal{B}_S \in \mathfrak{F}_S$ is chosen by SUs in $S$ with a certain probability $p_{\mathcal{B}_S}$ given by
\begin{align}\label{eq:proba}
p_{\mathcal{B}_S} = \begin{cases} \prod_{k \in \cup_{i=1}^{|S|}b_i,\ b_i \in \mathcal{B}_S} \theta_k \prod_{j \in \cup_{i=1}^{|S|}\mathcal{K}_{i,b_i}^S} (1-\theta_j), & \mbox{if } \cup_{i=1}^{|S|}b_i \cap \cup_{i=1}^{|S|}\mathcal{K}_{i,b_i}^S = \emptyset \\ 0,
\mbox{ otherwise;}
\end{cases}
\end{align}
where, for any SU $i\in S$, the set $\mathcal{K}_{i,b_i}^S =\{ j \in \mathcal{K}_i^S| \textrm{ rank}(j) < \textrm{rank}(b_i)\}$ represents the set of channels that need to be busy  before SU $i$ selects channel $b_i \in \mathcal{B}_S$, i.e., the set of channels ranked higher than $b_i$ (recall that the set $\mathcal{K}_i^S$ is ordered as a result of Algorithm~\ref{alg:sort}). If $\cup_{i=1}^{|S|}b_i \cap \cup_{i=1}^{|S|}\mathcal{K}_{i,b_i}^S \neq \emptyset$, it implies that, for the selection $\mathcal{B}_S$,  a channel needs to be available and busy at the same time which is impossible, and, hence, the probability of selecting any tuple  $\mathcal{B}_S \in \mathfrak{F}_S$ having this property is $0$. Due to this property, the SUs of any coalition $S \subseteq \mathcal{N}$, can only achieve a transmission capacity for the tuples $\mathcal{B}_S \in \bar{\mathfrak{F}}_S$ where $\bar{\mathfrak{F}}_S$ is the family of all \emph{feasible} tuples for coalition $S$ such that $\cup_{i=1}^{|S|}b_i \cap \cup_{i=1}^{|S|}\mathcal{K}_{i,b_i}^S = \emptyset$, which corresponds to the tuples which have a non-zero probability of occurrence as per (\ref{eq:proba}). %\footnote{The tuple corresponding to the case where no SU $i\in S$ finds a vacant channel has also a non-zero probability, but we omit it since its corresponding capacity is $0$.}.
Note that, the tuple corresponding to the case in which no SU $i\in S$ finds an unoccupied channel also has a non-zero probability, but is omitted as its corresponding capacity is $0$ and, thus, it has no effect on the utility.

For every channel selection $\mathcal{B}_S \in \bar{\mathfrak{F}}_S$, one can partition coalition $S$ into a number of \emph{disjoint} sets $\{S_1,\ldots,S_L\}$ with $\cup_{l=1}^{L} S_l = S$ such that, for a given $l \in \{1,\ldots,L\}$, the channels in $\mathcal{B}_S$ selected by any $i\in S_l$ are of the same rank. Thus, the SUs belonging to any $S_l$ access their selected channels \emph{simultaneously} and, for this reason, they can coordinate their channel access. In the event where $|S_l| =1$, the SU in $S_l$ simply transmits using its maximum power $\tilde{P}$ over its selected channel in $\mathcal{B}_S$. In contrast, for any $l \in \{1,\ldots,L\}$ with $|S_l| > 1$, the SUs in $S_l$ can share their sensing results (since they find their available channels simultaneously) and improve their access performance by distributing their powers cooperatively over the channels in the set $\mathcal{K}_{S_l}$ that corresponds to the channels selected by $S_l$ given  $\mathcal{B}_S $. With every SU $i\in S_l$, we associate a $1\times |\mathcal{K}_{S_l}|$ vector $\boldsymbol{P}_i^{\mathcal{B}_S} $ where each element $P_{i,k}^{\mathcal{B}_S }$ represents the power that SU $i \in S_l$ will use on channel $k \in \mathcal{K}_{S_l}$ given the selection $\mathcal{B}_S $. Let $\boldsymbol{P}_{ \mathcal{K}_{S_l}}^{\mathcal{B}_S} = [\boldsymbol{P}_1^{\mathcal{B}_S}\ \ldots \ \boldsymbol{P}^{\mathcal{B}_S}_{|S_l|}]^{T}$. Hence, for every $S_l,\ l\in\{1,\ldots,L\}$ such that  $|S_l| > 1$, the SUs can distribute their powers so as to maximize the total sum-rate that they achieve as a coalition, i.e., the social welfare, by solving \footnote{This choice allows us to capture both the selfish nature of the SUs (improving their individual utility using a competitive coalition formation process) and the cooperative nature of a coalition (in which the SUs act together for the overall benefit of the coalition using a fully cooperative social optimum at coalitional level). Other advanced optimization or game theoretic methods such as non-cooperative Nash equilibrium or Nash bargaining can also be used. However, these solutions can increase complexity and are out of the scope of this paper and will be addressed separately in future work.}:
\begin{align}\label{eq:opt}
\max_{\boldsymbol{P}_{ \mathcal{K}_{S_l}}^{\mathcal{B}_S} }\sum_{i \in S_l} \sum_{k \in  \mathcal{K}_{S_l}}  C_{i,k},
\end{align}
\begin{align*}
\textrm{ s.t. }
P_{i,k}^{\mathcal{B}_S} \ge 0,\ \forall i \in S_l, k \in \mathcal{K}_{S_l},\ \sum_{k \in \mathcal{K}_{S_l}} P_{i,k}^{\mathcal{B}_S} = \tilde{P},\ \forall i \in S_l,
\end{align*}
with $\tilde{P}$ the maximum transmit power and  $C_{i,k}$ the capacity achieved by SU $i \in S_l$ over channel $k \in \mathcal{K}_{S_l}$ and is given by
\begin{align}
 C_{i,k} =  \log{\left(1+\frac{P_{i,k}^{\mathcal{B}_S} \cdot g_{i,k}}{\sigma^2+I^{S_l}_{i,k} + I^{S\setminus S_l}_{i,k}+\bar{I}_{S,k} }\right)}
\end{align}
where $I^{S_l}_{i,k}  = \sum_{j \in S_l, j \neq i} g_{j,k} P_{j,k}^{\mathcal{B}_S}$ is the interference between SUs in $S_l$ on channel $k \in \mathcal{K}_{S_l}$, and $I^{S\setminus S_l}_{i,k}  = \sum_{j \in S\setminus S_l} g_{j,k} P_{j,k}^{\mathcal{B}_S}$ is the interference from SUs in $S \setminus S_l$ on channel  $k \in \mathcal{K}_{S_l}$ (if any). Further, $\bar{I}_{S,k}$ represents the average interference experienced by the members of coalition $S$, including SU $i$ from the SUs \emph{external} to $S$, which, given a partition $\Pi$ of $\mathcal{N}$ with $S \in \Pi$, corresponds to the SUs in $\mathcal{N} \setminus S$ (which can also be organized into coalitions as per $\Pi$). Similarly to the non-cooperative case, this average external interference can be estimated through measurements from the receiver (the receiver can inform every SU in $S$ of the interference it perceives, and then the SUs in $S$ can easily deduce the interference from the external sources).

Subsequently, given that, for any $\mathcal{B}_S \in \bar{\mathfrak{F}}_S$, $S$ is partitioned into $\{S_1,\ldots,S_L\}$ as previously described, the average capacity achieved, when acting cooperatively, by any SU $i \in S$, with $i \in S_l,\ l\in\{1,\ldots,L\}$ (for every $\mathcal{B}_S \in \bar{\mathfrak{F}}_S$),  is
\begin{align}\label{eq:capacityX}
\bar{C}_i^{S} = \sum_{\mathcal{B}_S \in \bar{\mathfrak{F}}_S} p_{\mathcal{B}_S} \cdot C_{i}^{\mathcal{B}_S},
\end{align}
where $p_{\mathcal{B}_S}$ is given by (\ref{eq:proba}), and $C_{i}^{\mathcal{B}_S}$ is the total capacity achieved by SU $i \in S_l$ when the SUs in $S$ select the channels in $\mathcal{B}_S$ and is given by
\begin{align}
C_{i}^{\mathcal{B}_S} =  \sum_{k \in \mathcal{K}_{S_l}}  C_{i,k}^{\mathcal{B}_S},
\end{align}
where $\mathcal{K}_{S_l} \subseteq \mathcal{K}_S$ is the set of channels available to $S_l \subseteq S$. Further, $ C_{i,k}^{\mathcal{B}_S}$ is the capacity achieved by SU $i \in S_l$ on channel $k\in \mathcal{K}_{S_l}$ given the channel selection $\mathcal{B}_S$  and is a direct result (upon computing the powers) of (\ref{eq:opt}) which is a standard constrained optimization problem that can be solved using well known methods \cite{BO00}.

Hence, the utility of any SU $i$ in coalition $S$ is given by
\begin{align}\label{eq:coop}
v_i(S,\Pi) = \bar{C}_i^{S} (1-\tau_i^S)
\end{align}
where $\Pi$ is the network partition currently in place which determines the external interference on coalition $S$, and $\tau_i^S$ is given by (\ref{eq:sensetime}) using the set $\mathcal{K}_i^S$ which is ordered by SU $i$, cooperatively with the SUs in $S$, using Algorithm~\ref{alg:sort}. Note that the utility in (\ref{eq:coop}) reduces to (\ref{eq:utilnc}) when the network is non-cooperative. Finally, we remark that, although cooperation can benefit the SUs both in the spectrum sensing and spectrum access levels, in many scenarios forming a coalition may also entail \emph{costs}. From a spectrum sensing perspective, due to the need for re-ordering the channels to reduce the interference, the sensing time of some members of a coalition may be longer than their non-cooperative counterparts. From a spectrum access perspective, by sharing information, some SUs may become subject to new interference on some channels (although reduced by the sorting algorithm) which may degrade their capacities. Thus, there exists a number of tradeoffs for cooperation, in different aspects for both sensing and access. In this regard, clearly, the utility in (\ref{eq:coop}) adequately captures these tradeoffs through the gains (or costs) in sensing time (spectrum sensing), and the gains (or costs) in capacity (spectrum access).

In a nutshell, with these tradeoffs, for maximizing their utilities in (\ref{eq:coop}), the SUs can cooperate to form coalitions, as illustrated in Fig.~\ref{fig:ill} for a network with $N=8$ and $K=10$. Subsequently, the next section provides an analytical framework to form coalitions such as in Fig.~\ref{fig:ill}.\vspace{-0.2cm}
\section{Joint Spectrum Sensing and Access as a Coalitional Game in Partition Form}\label{sec:game}
In this section, we cast the proposed joint spectrum sensing and access cooperative model as a coalitional game in partition form and we devise an algorithm for coalition formation.
\subsection{Coalitional Games in Partition Form: Basics}
For the purpose of deriving an algorithm that allows the SUs to form coalitions such as in Fig.~\ref{fig:ill} in a distributed manner, we use notions from cooperative game theory \cite{Game_theory2}. In this regard, denoting by $ \mathfrak{P}$ the set of all partitions of $\mathcal{N}$, we formulate the joint spectrum sensing and access model of the previous section as a coalitional game in \emph{partition form} with non-transferable utility which is defined as follows \cite{Game_theory2,WS00}:
\begin{definition}
A coalitional game in \emph{partition form} with \emph{non-transferable} utility~(NTU) is defined by a pair $(\mathcal{N},V)$ where  $\mathcal{N}$ is the set of players and $V$ is a mapping such that for every partition $\Pi \in  \mathfrak{P}$ , and every coalition $S \subseteq \mathcal{N},\ S\in \Pi$, $V(S,\Pi)$ is a closed convex subset of $\mathbb{R}^{|S|}$  that contains the payoff vectors that players in $S$ can achieve.
\end{definition}
\begin{figure}[e]
\begin{center}
\includegraphics[width=100mm]{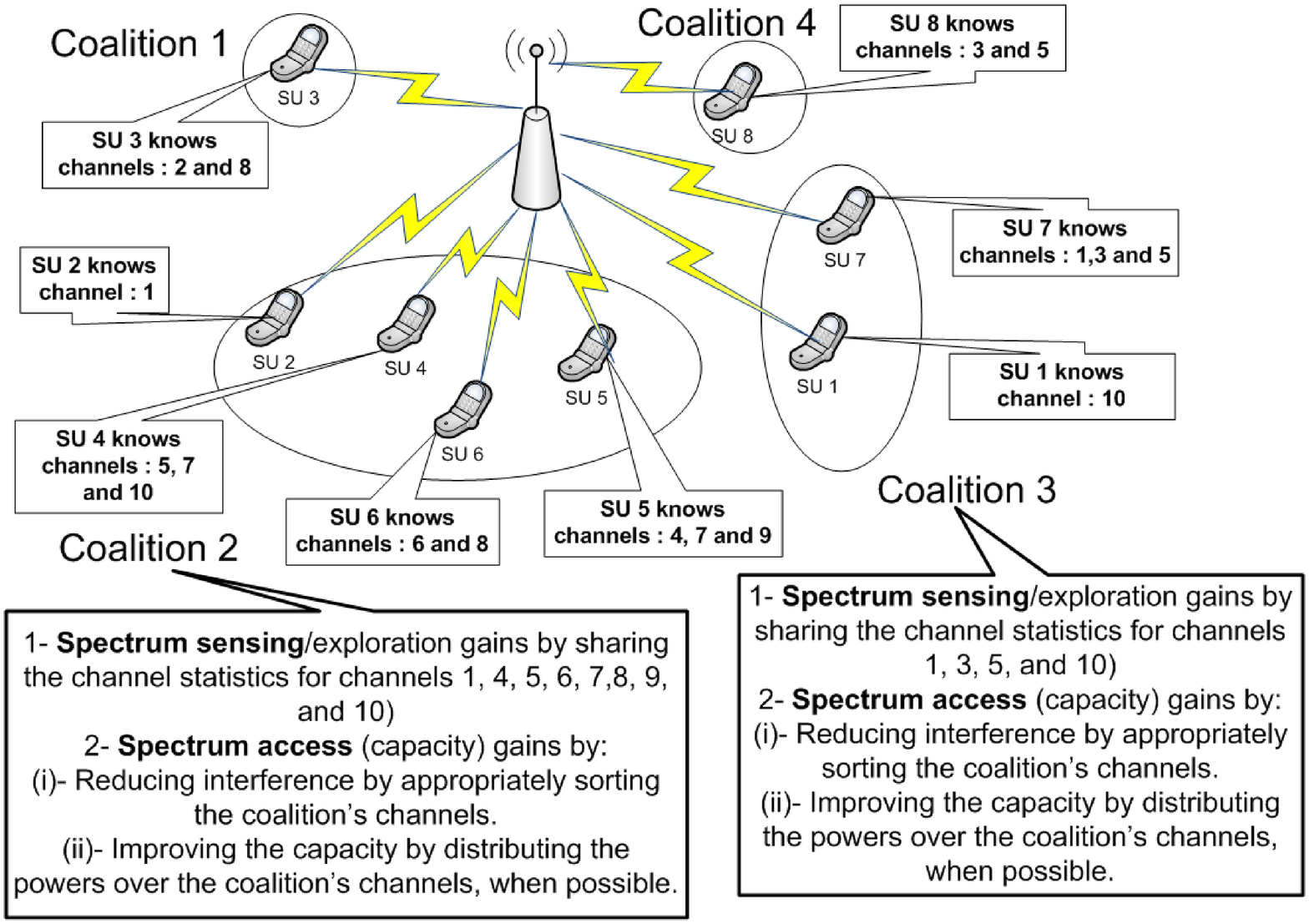}
\end{center}\vspace{-0.6cm}
\caption {An illustrative example of coalition formation for joint spectrum sensing and access for $N=8$ SUs and $K=10$ channels.} \label{fig:ill}
\end{figure}

Hence, a coalitional game is in partition form if, for any coalition $S \subseteq \mathcal{N}$, the payoff of every player in the coalition depends on the partition $\Pi$, i.e., on the players in $S$ as well as on the players in $\mathcal{N} \setminus S$. Further, the game has NTU if the utility received by $S$ cannot be expressed by a single value which can be arbitrarily divided among the coalition members, but is rather expressed as a set of vectors representing the payoffs that each member of $S$ can achieve when acting within $S$.

\subsection{Joint Sensing and Access as a Coalitional Game in Partition Form}
For the proposed joint spectrum sensing and access problem, given a partition $\Pi$ of $\mathcal{N}$ and a coalition $S \in \Pi$, and denoting by $x_i(S,\Pi)$ the payoff of SU $i \in S$ received when acting in coalition $S$ when $\Pi$ is in place, we define the coalitional value set, i.e., the mapping $V$ as follows:
\begin{align}\label{eq:utilcoal}
V(S,\Pi)=\{\boldsymbol{x}(S,\Pi) \in \mathbb{R}^{|S|}|\forall i \in S,x_i(S,\Pi) = v_i(S,\Pi)\},
\end{align}
where $v_i(S,\Pi)$ is given by (\ref{eq:coop}). Using (\ref{eq:utilcoal}), we note:
\begin{remark}
The proposed joint spectrum sensing and access game can be modeled as a $(\mathcal{N},V)$ coalitional game in partition form with non-transferable utility where the mapping $V$ is a singleton set as given by (\ref{eq:utilcoal}), and, hence, is a closed and convex subset of  $\mathbb{R}^{|S|}$.
\end{remark}

Coalitional games in partition form have  recently attracted attention in game
theory~\cite{Game_theory2,DR00,KA01,WS00,DM00,Lazlo1,Lazlo2}. Partition form games are characterized by the dependence of the payoffs on externalities, i.e., on the way the network is partitioned. Unlike coalitional games in characteristic form for which the focus is on studying the stability of the grand coalition of all players \cite{Game_theory2}, games in partition form provide a richer and more complex framework since any coalitional structure can be optimal \cite{DR00,WS00,Lazlo1,Lazlo2}. In this regard, coalitional games in partition form are often classified as \emph{coalition formation} games \cite{WS00}. Hence, traditional solution concepts for coalitional games, such as core-stable partitions or the Shapley value\cite{Game_theory2}, are inapplicable to coalitional games in partition form \cite{Game_theory2,DR00,WS00,Lazlo1,Lazlo2}. For instance, for coalition formation games in partition form, there is a need for devising algorithms to form the coalitional structure that can potentially emerge in the network. In particular, for the proposed joint spectrum sensing and access coalitional game, due to the tradeoffs between the benefits and costs of cooperation as captured by (\ref{eq:coop}) and explained in Section~\ref{sec:jssa}, we note the following:
\begin{remark}
In the proposed joint spectrum sensing and access $(\mathcal{N},V)$  coalitional game in partition form, due to the dependence on externalities and the benefit-cost tradeoffs from cooperation as expressed in (\ref{eq:coop}) and (\ref{eq:utilcoal}), any coalitional structure may form in the network and the grand coalition is seldom beneficial due to increased costs. Hence, the proposed joint sensing and access game is classified as a coalition formation game in partition form.
\end{remark}

\subsection{Proposed Preference Relations}
Most coalition formation algorithms in the game theory literature \cite{DR00,WS00} are built for games in characteristic form. Although some approaches for the partition form are presented in \cite{DR00,Lazlo1,Lazlo2}, most of these are targeted at solving problems in economics with utilities quite different from the one dealt with in this paper. Moreover, the approaches using the recursive core in \cite{Lazlo1,Lazlo2} (which are an extension to the classical characteristic form solutions such as core-stable partitions) rely heavily on combinatorial techniques which are unsuitable for wireless problems such as in cognitive radio. In order to build a coalition formation algorithm suitable for joint spectrum sensing and access, we borrow concepts from \cite{HC00}, in which the players build coalitions based on preferences, and extend them to accommodate the partition form.

\begin{definition}
For any SU $i\in \mathcal{N}$, a \emph{preference relation} or \emph{order} $\succeq_i$ is defined as a complete, reflexive, and transitive binary relation over the set of all coalition/partition pairs that SU $i$ can be a member of, i.e., the set $\{(S_k,\Pi) | S_k \subseteq \mathcal{N},\ i \in S_k,\ S_k \in \Pi, \ \Pi \in \mathfrak{P}\}$.
\end{definition}

Consequently, for any SU $i \in \mathcal{N}$, given two coalitions and their respective partitions $S_1 \subseteq \mathcal{N},\ S_1 \in \Pi$ and, $S_2 \subseteq \mathcal{N}, \  S_2\in \Pi^{\prime}$ such that $i \in S_1$ and $i \in S_2$,  $(S_1,\Pi) \succeq_i (S_2,\Pi^{\prime})$ indicates that player $i$ prefers to be part of coalition $S_1$ when $\Pi$ is in place, over being part of coalition $S_2$ when $\Pi^{\prime}$ is in place, or at least, $i$ prefers both coalition/partition pairs equally. Further, using the asymmetric counterpart of $\succeq_i$, denoted by $\succ_i$, then $(S_1,\Pi) \succ_i (S_2,\Pi^{\prime})$, indicates that player $i$ \emph{strictly} prefers being a member of $S_1$ within $\Pi$ over being a member of $S_2$ with $\Pi^{\prime}$. We also note that the preference relation can be used to compare two coalitions in the same partition, or the same coalition in two different partitions.

For every application, an adequate preference relation $\succeq_i$ can be defined to allow the players to quantify their preferences depending on their parameters of interest. In this paper, we propose the following preference relation for any SU $i\in \mathcal{N}$:
\begin{align}\label{eq:prefsu}
(S_1,\Pi) \succeq_i (S_2,\Pi^{\prime}) \Leftrightarrow \phi_i(S_1,\Pi) \ge \phi_i(S_2,\Pi^{\prime})
\end{align}
where $S_1 \in \Pi,\ S_2 \in \Pi^{\prime}$, with $\Pi, \Pi^{\prime} \in \mathfrak{P}$, are any two coalitions that contain SU $i$, i.e., $i \in S_1$ and $i \in S_2$ and $\phi_i$ is a preference function defined for any SU $i\in \mathcal{N}$ as follows ($S$ is a coalition containing $i$)
\begin{align}\label{eq:pref1}
\phi_i(S,\Pi) = \begin{cases} x_i(S,\Pi), & \mbox{if } \left(x_j(S,\Pi) \ge x_j(S \setminus \{i\},\Pi),\forall j \in S\setminus\{i\}\ \&\ S \notin h(i)\right) \mbox { or } (|S|=1)  \\ 0, &\mbox{otherwise}, \end{cases}
\end{align}
where $x_i(S,\Pi)$ is given by (\ref{eq:coop}) through (\ref{eq:utilcoal}) and it represents the payoff received by SU $i$ in coalition $S$ when partition $\Pi$ is in place, and $h(i)$ is the history set of SU $i$ which is a set that contains the coalitions of size larger than $1$ that SU $i$ was member of (visited) in the past, and has parted.

The main rationale behind the preference function $\phi_i$ is that any SU $i$ assigns a preference equal to its achieved payoff for any coalition/partition pair $(S,\Pi)$ such that either: (i)- $S$ is the singleton coalition, i.e., SU $i$ is acting non-cooperatively, or (ii)- the presence  of SU $i$ in coalition $S$ is not detrimental to any of the SUs in $S \setminus \{i\}$, and coalition $S$ has not been previously visited by SU $i$, i.e., is not in the history $h(i)$. Otherwise, the SU assigns a preference value of $0$ to any coalition whose members' payoffs decrease due to the presence of $i$, since such a coalition would refuse to have $i$ join the coalition. Also, any SU $i$ assigns a preference of $0$ to any coalition that it has already visited in the past and has \emph{left} since an SU $i$ has no incentive to revisit a coalition it has previously left due to a decrease in the utility.

Having defined the main ingredients of the proposed game, in the next subsection, we devise an algorithm for coalition formation.
\subsection{Proposed Algorithm: Coalition Formation Rule and Algorithm Phases}\label{sec:coalalg}
In order to devise a coalition formation algorithm based on the SUs' preferences, we propose the following rule:
\begin{definition}\label{def:switch}
\textbf{Switch Rule -} Given a partition $\Pi=\{S_1,\ldots,S_M\}$ of the set of SUs $\mathcal{N}$, an SU $i$ decides to leave its current coalition $S_m,\ $ for some $m \in \{1,\ldots,M\}$ and join another coalition $S_k \in \Pi \cup \{\emptyset\},\ S_k \neq S_m$, hence forming $\Pi^{\prime} = \{\Pi \setminus \{S_m,S_k\}\} \cup \{S_m\setminus\{i\},S_k\cup\{i\}\}$, if and only if $(S_k \cup \{i\},\Pi^{\prime}) \succ_i (S_m,\Pi)$. Hence, $\{S_m,S_k\} \rightarrow \{S_m\setminus\{i\},S_k\cup\{i\}\}$ and $\Pi \rightarrow \Pi^{\prime}$.
\end{definition}

For any partition $\Pi$, the switch rule provides a mechanism whereby any SU can leave its current coalition $S_m$ and join another coalition $S_k \in \Pi$, forming a new partition $\Pi^{\prime}$, given that the new pair $(S_k \cup \{i\},\Pi^{\prime})$ is strictly preferred over $(S_m,\Pi)$ through the preference relation defined by (\ref{eq:prefsu}) and (\ref{eq:pref1}). That is, an SU would \emph{switch} to a new coalition if it can strictly improve its payoff, \emph{without} decreasing the payoff of any member of the new coalition. Thus, the switch rule can be seen as an individual decision made by an SU, to move from its current coalition to a new coalition while improving its payoff, given the \emph{consent} of the members of this new coalition as per (\ref{eq:prefsu}). Further, whenever an SU decides to switch from its current coalition $S_m \in \Pi$ to join a different coalition, coalition $S_m$ is stored in its history set $h(i)$ (if $|S_m| >1$). %Hence, given a partition $\Pi$, whenever an SU $i$ decides to leave coalition $S_m \in \Pi$ to join a different coalition, coalition $S_m$ is automatically stored by SU $i$ in its history set $h(i)$.

Consequently, we propose a coalition formation algorithm composed of three main phases: Neighbor discovery, coalition formation, and joint spectrum sensing and access. In the first phase, the SUs explore neighboring SUs (or coalitions) with whom they may cooperate. For discovering their neighbors, neighbor discovery algorithms suited for cognitive radio such as those in \cite{ND00} and \cite{ND01} may be used. Once neighbor discovery is complete, the next phase of the algorithm is the coalition formation phase. First, the SUs start by investigating the possibility of performing a switch operation by engaging in pairwise negotiations with discovered SUs/coalitions. Once an SU identifies a potential switch operation (satisfying (\ref{eq:prefsu}) and (\ref{eq:pref1})), it can make a \emph{distributed} decision to switch and join a new coalition. In this phase,  we consider that, the order in which the SUs make their switch operations is random but sequential (dictated by who requests first to cooperate). For any SU, a switch operation is easily performed as the SU can leave its current coalition and join the new coalition whose members already agree on the joining of this SU as per (\ref{eq:prefsu}) and (\ref{eq:pref1}).

\subsection{Convergence and Properties of the Proposed Algorithm}
The convergence of the proposed coalition formation algorithm during this phase is guaranteed as follows:
\begin{theorem}\label{th:one}
Starting from any initial network partition $\Pi_{\text{init}}$, the coalition formation phase of the proposed algorithm always converges to a final network partition $\Pi_f$ composed of a number of disjoint coalitions of SUs.
\end{theorem}
\begin{proof}
Denote by $\Pi_{n_{l,i}}^{l,i}$ the partition formed at iteration $l$ during the time SU $i\in \mathcal{N}$ needs to act after the occurrence of $n_{l,i}$ switch operations by one or more SUs up to the turn of SU $i$ in iteration $l$. Consider that the SUs act in ascending order, i.e., SU $1$ acts first, then SU $2$, and so on. Given any initial starting partition $\Pi_{\textrm{init}}=\Pi^{1,1}_{0}$, the coalition formation phase of the proposed algorithm consists of a sequence of switch operations as follows (as an example):
\begin{align}\label{eq:trans}
\Pi^{1,1}_{0}\rightarrow \Pi^{1,2}_{1} \rightarrow \ldots \rightarrow \Pi^{1,N}_{n_{1,N}} \ldots \rightarrow \Pi^{l,N}_{n_{l,N}} \rightarrow \ldots,
\end{align}
where the operator $\rightarrow$ indicates a switch operation. Based on (\ref{eq:prefsu}), for any two partitions  $\Pi_{n_{l,i}}^{l,i}$ and $\Pi_{n_{m,j}}^{m,j}$ in (\ref{eq:trans}), such that $n_{l,i} \neq n_{m,j}$, i.e., $\Pi_{n_{m,j}}^{m,j}$ is a result of the transformation of $\Pi_{n_{l,i}}^{l,i}$ (or vice versa) after a number of switch operations, we have two cases: (C1)- $\Pi_{n_{l,i}}^{l,i} \neq \Pi_{n_{m,j}}^{m,j}$, or (C2)- an SU revisited its non-cooperative state, and thus $\Pi_{n_{l,i}}^{l,i} = \Pi_{n_{m,j}}^{m,j}$.

If (C1) is true for all $i,k \in \mathcal{N}$ for any two iterations $l$ and $m$, and, since the number of partitions of a set is \emph{finite} (given by the Bell number \cite{DR00}), then the number of transformations in (\ref{eq:trans}) is finite. Hence, in this case,  the sequence in (\ref{eq:trans}) will always terminate after $L$ iterations and converge to a final partition $\Pi_f = \Pi^{L,N}_{n_{L,N}}$ (without oscillation). If case (C2) also occurs in (\ref{eq:trans}),  future switch operations (if any) for any SU that reverted to act non-cooperatively will always result in a new partition as per (\ref{eq:prefsu}). Thus, even when (C2) occurs, the finite number of partitions guarantees the algorithm's convergence to some $\Pi_f$. Hence, the coalition formation phase of the proposed algorithm always converges to a final partition $\Pi_f$.
\end{proof}

The stability of the partition $\Pi_f$ resulting from the convergence of the proposed algorithm can be studied using the following stability concept (modified from \cite{HC00} to accommodate the partition form):
\begin{definition}
A partition $\Pi = \{S_1,\ldots,S_M\}$ is  \emph{Nash-stable} if $\forall i \in \mathcal{N} \textrm{ such that } i\in S_m, S_m \in \Pi,$ we have $(S_m,\Pi) \succeq_i (S_k \cup \{i\},\Pi^{\prime})$ for all $S_k \in \Pi \cup \{\emptyset\}$ with $\Pi^{\prime} = (\Pi \setminus \{S_m,S_k\} \cup \{S_m\setminus\{i\},S_k\cup\{i\}\})$.
\end{definition}
Hence, a partition $\Pi$ is Nash-stable if no SU has an incentive to move from its current coalition to another coalition in $\Pi$ or to deviate and act alone. %Note that, the above definition requires Nash stability using the strict preference relation $\succ_i$ unlike the looser definition in \cite{HC00} which is based on the preference relation $\succeq_i$.

\begin{proposition}\label{prop:one}
Any partition $\Pi_f$ resulting from the coalition formation phase of the proposed algorithm is Nash-stable.
\end{proposition}
\begin{proof}
 If the partition $\Pi_f$ resulting from the proposed algorithm is \emph{not} Nash-stable then, there $\exists i \in \mathcal{N}$ with $i \in S_m,\ S_m\in \Pi_f$, and a coalition $S_k \in \Pi_f$ such that $(S_k \cup \{i\},\Pi^{\prime}) \succ_i (S_m,\Pi)$, and hence, SU $i$ can perform a \emph{switch} operation which contradicts with the fact that $\Pi_f$ is the result of the convergence of the proposed algorithm (Theorem~\ref{th:one}). Thus, any partition $\Pi_f$ resulting from the coalition formation phase of the proposed algorithm is Nash-stable.
\end{proof}

Following the convergence of the coalition formation phase to a Nash-stable partition, the third and last phase of the algorithm entails the joint spectrum sensing and access where the SUs operate  using the model described in Section~\ref{sec:jssa}  for locating unoccupied channels and transmitting their data cooperatively. A summary of one round of the proposed algorithm is given in Algorithm~\ref{alg:coalform}. The proposed algorithm can adapt the coalitional structure to environmental changes such as a change in the PU traffic or slow channel variations (e.g., due to slow mobility). For this purpose, the first two phases of the algorithm shown in Algorithm~\ref{alg:coalform} are repeated periodically over time, allowing the SUs, in Phase~2, to take distributed decisions to adapt the network's topology through new switch operations (which  would converge independent of the starting partition as per Theorem~\ref{th:one}). Thus, for time varying environments, every period of time $\eta$ the SUs assess whether it is possible to switch from their current coalition. Note that the history set $h(i)$ for any SU $i \in \mathcal{N}$ is also reset every $\eta$ time units. %The period $\eta$ is generally smaller in highly varying environments to allow for a more adequate adaptation of the topology.

\subsection{Implementation Issues}
The proposed algorithm can be implemented in a distributed way, since, as already explained, the switch operation can be performed by the SUs independently of any centralized entity. First, for neighbor discovery, the SUs can either utilize existing algorithms such as those in \cite{ND00} and \cite{ND01}, or they can rely on information from control channels such as the recently proposed cognitive pilot channel~(CPC) which provides frequency, location, and other information for assisting the SUs in their operation \cite{CPC01,CPC02}. Following neighbor discovery, the SUs engage in pairwise negotiations, over control channels, with their neighbors. In this phase, given a present partition $\Pi$, for each SU, the computational complexity of finding its next coalition, i.e., locating a switch operation, is easily seen to be $O(|\Pi|)$ in the worst case, and the largest value of $|\Pi|$ occurs when all the SUs are non-cooperative, in which case $|\Pi| = N$. Clearly, as coalitions start to form, the complexity of locating a potential switch operation becomes smaller. Also, for performing a switch, each SU and coalition have to evaluate their potential utility through (\ref{eq:coop}), to determine whether a switch operation is possible. For doing so, the SUs need to know the external interference and to find all feasible permutations to compute their average capacities. Each SU in the network is made aware of the average external interference it experiences through measurements fed back from the receiver to the SU. As a result, for forming a coalition, the SUs compute the average external interference on the coalition by combining their individual measurements. Alternatively, for performing coalition formation, the SUs can also rely on information from the CPC which can provide a suitable means for gathering information on neighbors and their transmission schemes. Moreover, although, at first glance, finding all feasible permutations may appear complex, as per Section~\ref{sec:jssa}, the number of feasible permutations is generally small with respect to the total number of permutations due to the condition in (\ref{eq:proba}). Further, as cooperation entails costs, the network eventually deals with small coalitions (as will be seen in Section~\ref{sec:sim}) where finding these feasible permutations will be reasonable in complexity.
  \begin{algorithm}[e]
\caption{\footnotesize One round of the proposed coalition formation algorithm}
\label{alg:coalform}
\begin{algorithmic}
\footnotesize
\STATE \textbf{Initial State}
\STATE The network is partitioned by $\Pi_{\textrm{initial}}=\{S_1,\ldots,S_M\}$. At the beginning of all time, the network is non-cooperative, hence, $\Pi_{\textrm{init}}=\mathcal{N}$.\vspace{0.2em}
\STATE \hspace*{1em}\textbf{Phase 1 - Neighbor Discovery:}
\STATE \hspace*{1.5em}Each SU in $\mathcal{N}$ surveys its neighborhood for existing coalitions, \STATE \hspace*{1.5em}in order to learn the partition $\Pi$ in place using existing
\STATE \hspace*{1.5em}neighbor discovery algorithms such as in \cite{ND00,ND01}.
\STATE \hspace*{1em}\textbf{Phase 2 - Coalition Formation:}
\REPEAT
\STATE Each SU $i\in \mathcal{N}$ investigates potential switch operations using the preference in (\ref{eq:prefsu}) by engaging in pairwise negotiations with existing coalitions in partition $\Pi$ (initially $\Pi = \Pi_{\textrm{init}})$.
\STATE \hspace*{1.5em}Once a switch operation is found:
\STATE \hspace*{2em}a) SU $i$ leaves its current coalition.
\STATE \hspace*{2em}b) SU $i$ updates its history $h(i)$, if needed.
\STATE \hspace*{2em}c) SU $i$ joins the new coalition with the consent of its members.
\UNTIL{convergence to a Nash-stable partition}
\STATE \hspace*{1em}\textbf{Phase 3 - Joint Spectrum Sensing and Access:}
\STATE \hspace*{1.5em}The formed coalitions perform joint cooperative spectrum
\STATE \hspace*{1.5em}sensing and access  as per Section~\ref{sec:jssa}.
\STATE \textbf{By periodic runs of these phases, the algorithm allows the SUs to adapt the network structure to environmental changes (see Section~\ref{sec:coalalg}).}
\end{algorithmic}\vspace{-0.1cm}
\end{algorithm}\vspace{-0.55cm}

 %Finally, in changing environments, as the algorithm is repeated periodically and since we consider only periodic changes, the complexity of the coalition formation algorithm is comparable to the one in the static environment, but with more runs of the algorithm

\section{Simulation Results and Analysis}\label{sec:sim}
\subsection{Simulation Parameters and Sample Network Snapshot}
In order to simulate our proposed approach, we setup a system-level simulator in MATLAB as follows: The BS is
placed at the origin of a $3$km $\times 3$km square area with the SUs randomly deployed in the area around it. We set the maximum SU transmit power to $\tilde{P}=10$~mW, the noise variance to $\sigma^2=-90$~dBm, and the path loss exponent to $\mu=3$. Unless stated otherwise, we set the fraction of time for sensing a single channel to $\alpha=0.05$ and we consider networks with $K=14$~channels\footnote{As an example, this can map to the total channels in 802.11b, although the actual used number varies by region ($11$ for US, $13$ for parts of Europe, etc.) \cite{RAP00}.}. In addition, non-cooperatively, we assume that each SU can accurately learn the statistics of $K_i=3$ channels, $\forall i \in \mathcal{N}$ (for every SU $i$ these non-cooperative $K_i$~channels are randomly chosen among the available PUs\footnote{This method of selection is considered as a general case, other methods for non-cooperatively choosing the PU channels can also be accommodated.}).

\begin{figure}[e]
\begin{center}
\includegraphics[width=10cm]{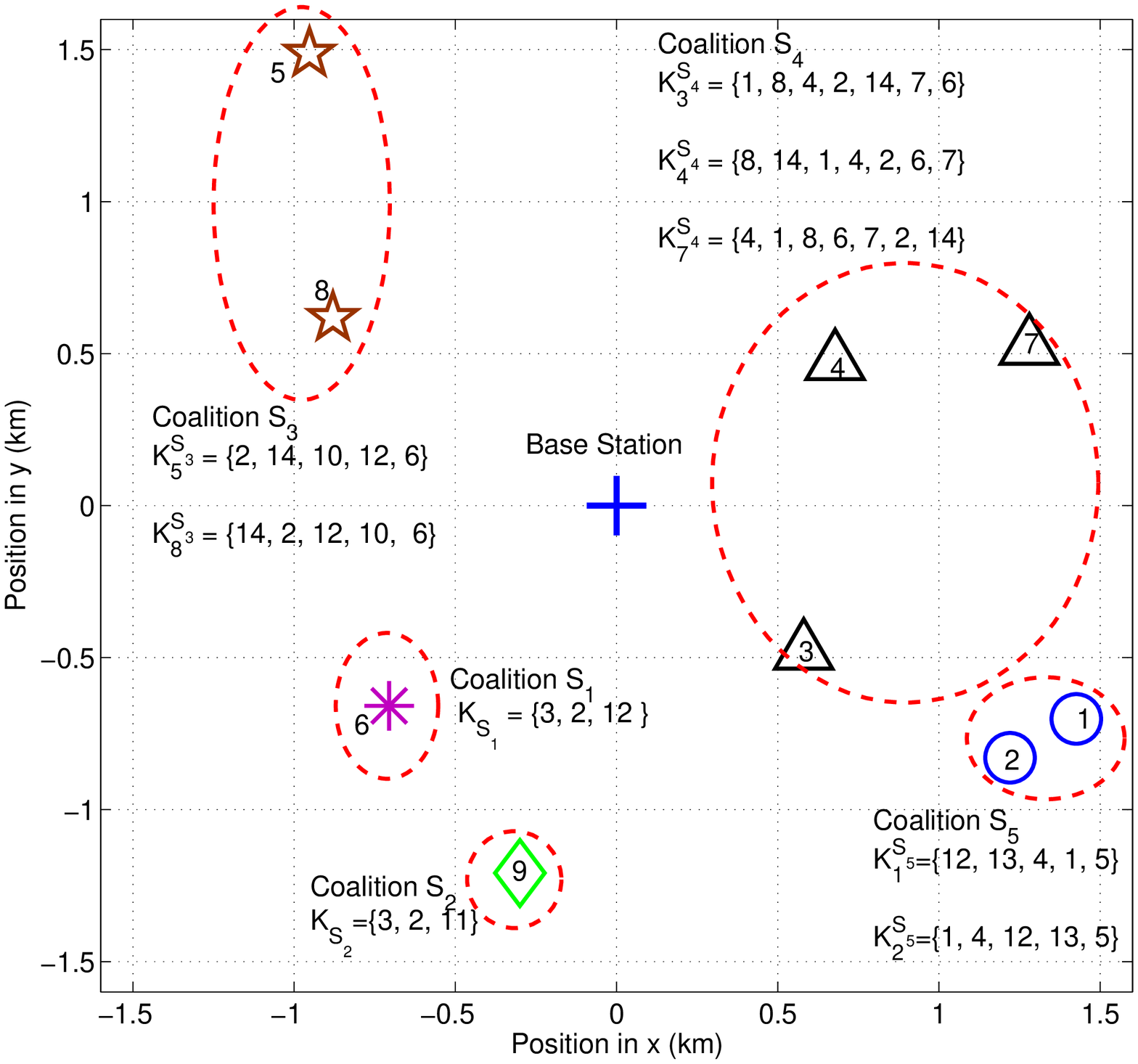}
\end{center}\vspace{-1cm}
\caption {A snapshot of a network partition resulting from the proposed algorithm with $N=9$~SUs and $K=14$~channels.} \label{fig:snap}
\end{figure}

 Fig.~\ref{fig:snap} shows a snapshot of the network structure resulting from the proposed coalition formation algorithm for a randomly deployed network with $N=9$~SUs and $K=14$ channels. The probabilities that the channels are unoccupied are: $\theta_1=0.98$, $\theta_2=0.22$, $\theta_3=0.64$, $\theta_4=0.81$, $\theta_5=0.058$, $\theta_6=0.048$, $\theta_7=0.067$, $\theta_8=0.94$, $\theta_9=0.18$, $\theta_{10}=0.25$, $\theta_{11}=0.17$, $\theta_{12}=0.15$,  $\theta_{13}=0.23$, $\theta_{14}=0.36$. In Fig.~\ref{fig:snap}, the SUs self-organize into $5$ coalitions forming partition $\Pi_f=\{S_1, S_2, S_3, S_4, S_5\}$. For each coalition in $\Pi_f$,  Fig.~\ref{fig:snap} shows the \emph{sorted} (by Algorithm~\ref{alg:sort}) set of channels used by the SUs in the coalitions (note that channel $9$ was not learned by any SU non-cooperatively). By inspecting the channel sets used by $S_3, S_4$, and $S_5$, we note that, by using Algorithm~1 the SUs sort their channels in a way to avoid selecting the same channel at the same rank, when possible. This is true for all ranks of these coalitions with two exceptions: The last rank for coalition $S_3$ where SUs $5$ and $8$ both rank channel $6$ last since it is rarely available as $\theta_6 = 0.048$, and, similarly, the last rank for coalition $S_5$ where SUs $1$ and $2$ both select channel $5$ (ranked lowest by both SUs) since it is also seldom available as  $\theta_5 = 0.058$. The partition $\Pi_f$ in Fig.~\ref{fig:snap} is Nash-stable, as no SU has an incentive to change its coalition. For example, the non-cooperative utility of SU $9$ is $x_9(\{9\},\Pi_f)=1.1$, by joining with SU $6$, this utility drops to $0.38$, also, the utility of SU $6$ drops from $x_6(\{6\},\Pi_f)=1.79$ to $1.63$. This result shows that cooperation can entail a cost, notably, due to the fact that
that both SUs $6$ and $9$ know, non-cooperatively, almost the same channels (namely, $3$ and $2$), and hence, by cooperating they suffer a loss in sensing time which is not compensated by the access gains. Due to the cooperation tradeoffs, the utility of SU $9$, drops to $0.797$, $0.707$, and $0.4624$, if SU $9$ joins coalitions $S_3$, $S_4$, or $S_5$, respectively. Thus, SU $9$ has no incentive to switch its current coalition. This property can be verified for all SUs in Fig~\ref{fig:snap} by inspecting the variation of their utilities if they switch their coalition and, thus, partition $\Pi_f$ is Nash-stable. %Briefly, Fig~\ref{fig:snap} provides an insight on how SUs' coalition formation can take place for joint spectrum sensing and access given the cooperation tradeoffs.

\subsection{Performance Assessment}
In Fig.~\ref{fig:perf}, we show the average payoff achieved per SU per slot for a network with $K=14$ channels as the number of SUs, $N$, in the network increases. The results are averaged over random positions of the SUs and the random realizations of the probabilities $\theta_k, \forall k \in \mathcal{K}$. The performance of the proposed algorithm is compared with the classical non-cooperative scheme as well as with the optimal centralized solution found using an exhaustive search. Note that beyond $N=8$~SUs finding the optimal solution becomes mathematically and computationally intractable as the number of partitions increases exponentially with $N$ as per the Bell number~\cite{DR00}. Fig.~\ref{fig:perf} shows that, as the number of SUs $N$ increases, the performance of all three schemes decreases due to the increased interference. However, at all network sizes, the proposed coalition formation algorithm maintains a better performance compared to the non-cooperative case. In fact, the proposed joint spectrum sensing and access presents a significant performance advantage over the non-cooperative case, increasing with $N$ as the SUs are more likely (and willing, due to increased interference) to find cooperating partners when $N$ increases. This performance advantage reaches up to $86.8\%$ relative to the non-cooperative case at $N=20$~SUs. Further, Fig.~\ref{fig:perf} shows that the optimal solution has a $23.1\%$ advantage over the proposed scheme at $N=4$~SUs, but this advantage decreases to around $19.9\%$ at $N=8$~SUs. This result indicates that the performance of the Nash-stable partitions resulting from the proposed algorithm becomes closer to the optimal solution as the number of SUs $N$ increases. This implies that, as more partners become available for selection, the proposed algorithm can reach a more efficient Nash-stable partition.

\begin{figure}[e]
\begin{center}
\includegraphics[width=100mm]{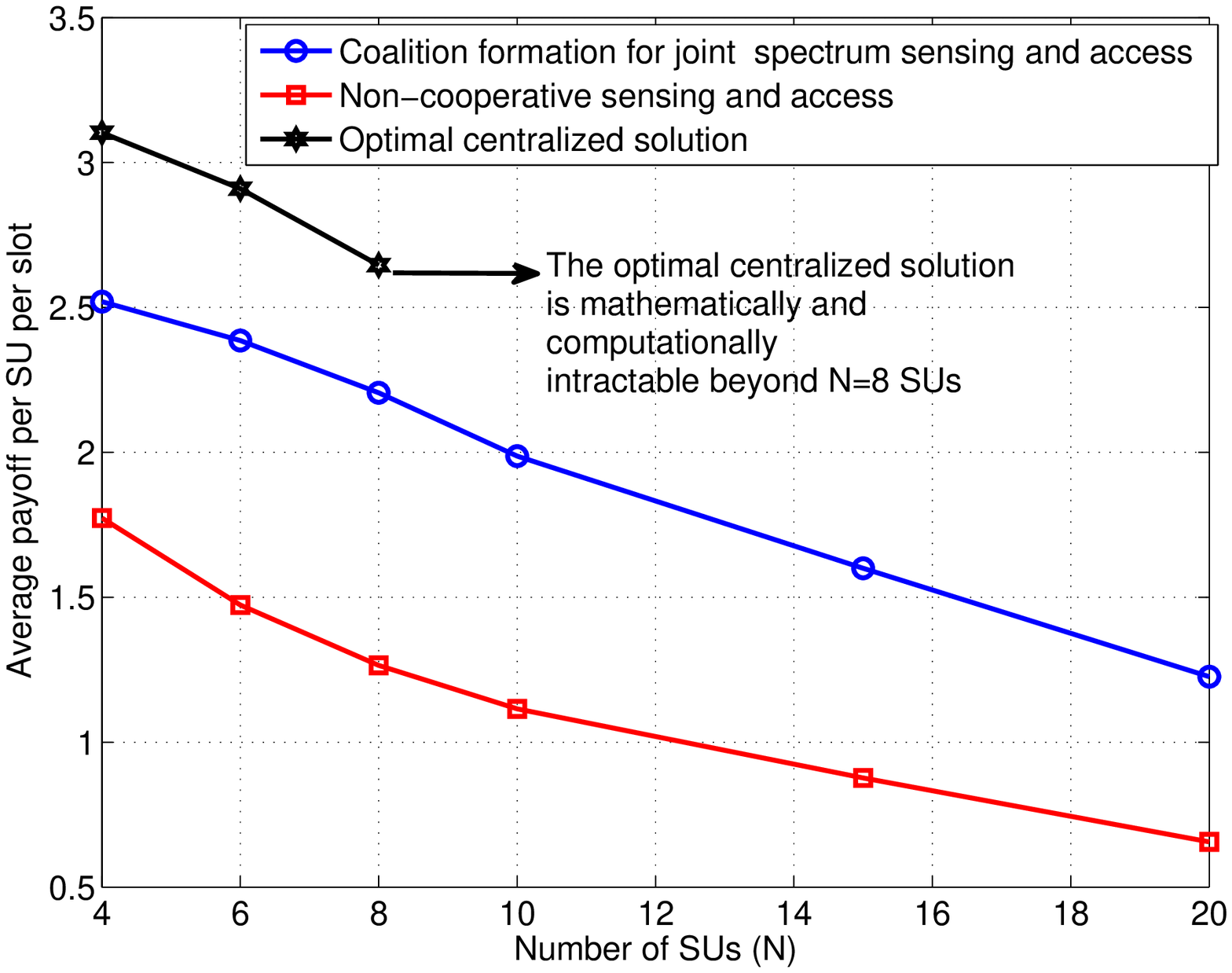}
\end{center}\vspace{-0.8cm}
\caption {Average payoff achieved per SU per slot (averaged over random positions of the SUs and the random realizations of the probabilities $\theta_k, \forall k \in \mathcal{K}$) for a network with $K=14$ channels as the network size $N$ varies.} \label{fig:perf}
\end{figure}
\begin{figure}[e]
\begin{center}
\includegraphics[width=100mm]{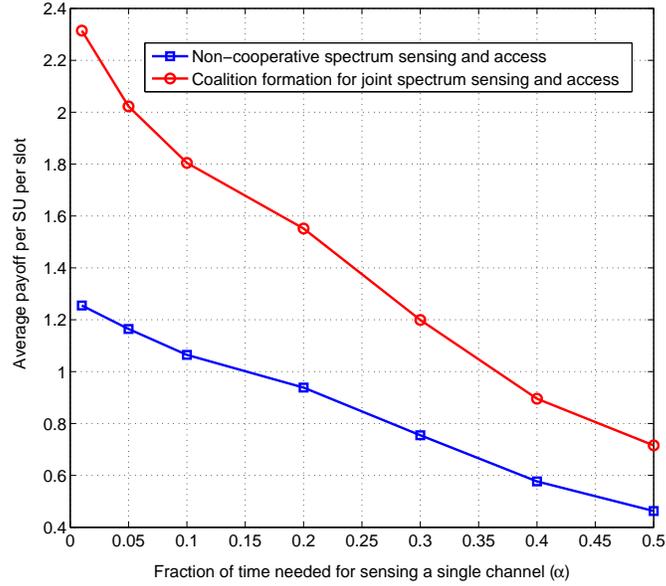}
\end{center}\vspace{-0.8cm}
\caption {Average payoff achieved per SU per slot (averaged over random positions of the SUs and the random realizations of the probabilities $\theta_k, \forall k \in \mathcal{K}$) for a network with $N=10$~SUs and $K=14$~channels as the fraction of time needed for sensing a single channel $\alpha$ varies.} \label{fig:perfalpha}
\end{figure}
\begin{figure}[e]
\begin{center}
\includegraphics[width=100mm]{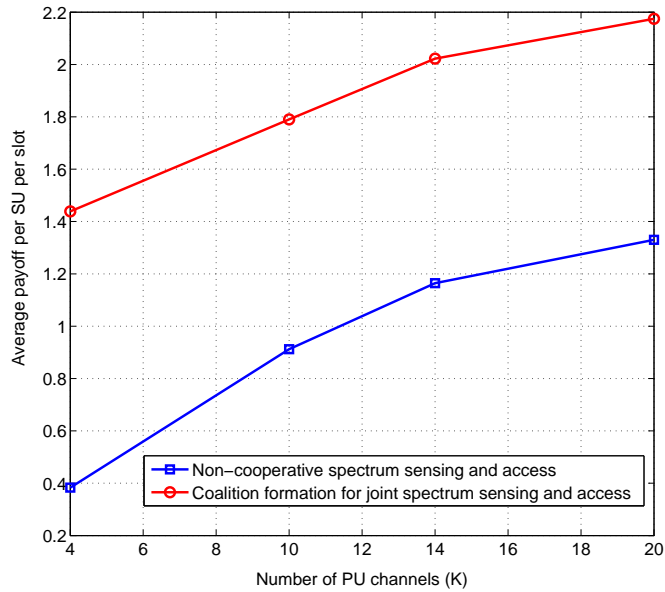}
\end{center}\vspace{-0.8cm}
\caption {Average payoff achieved per SU per slot (averaged over random positions of the SUs and the random realizations of the probabilities $\theta_k, \forall k \in \mathcal{K}$) for a network with $N=10$~SUs as the number of channels $K$ varies.} \label{fig:perfchan}
\end{figure}

\begin{figure}[e]
\begin{center}
\includegraphics[width=100mm]{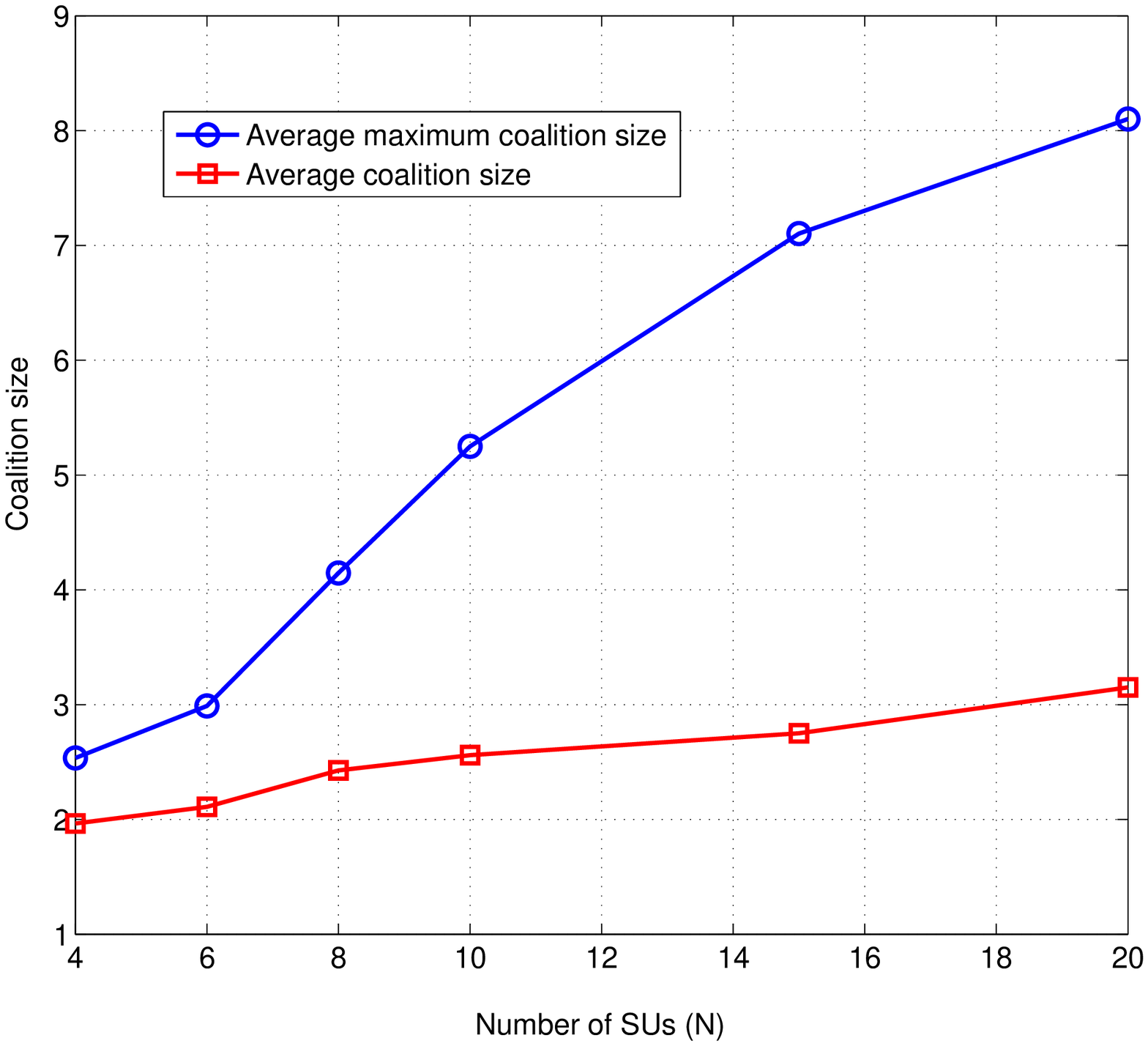}
\end{center}\vspace{-0.8cm}
\caption {Average and average maximum coalition size (averaged over random positions of the SUs and the random realizations of the probabilities $\theta_k, \forall k \in \mathcal{K}$) for a network with $K=14$ channels as the network size $N$ varies.} \label{fig:size}
\end{figure}

\begin{figure}[e]
\begin{center}
\includegraphics[width=100mm]{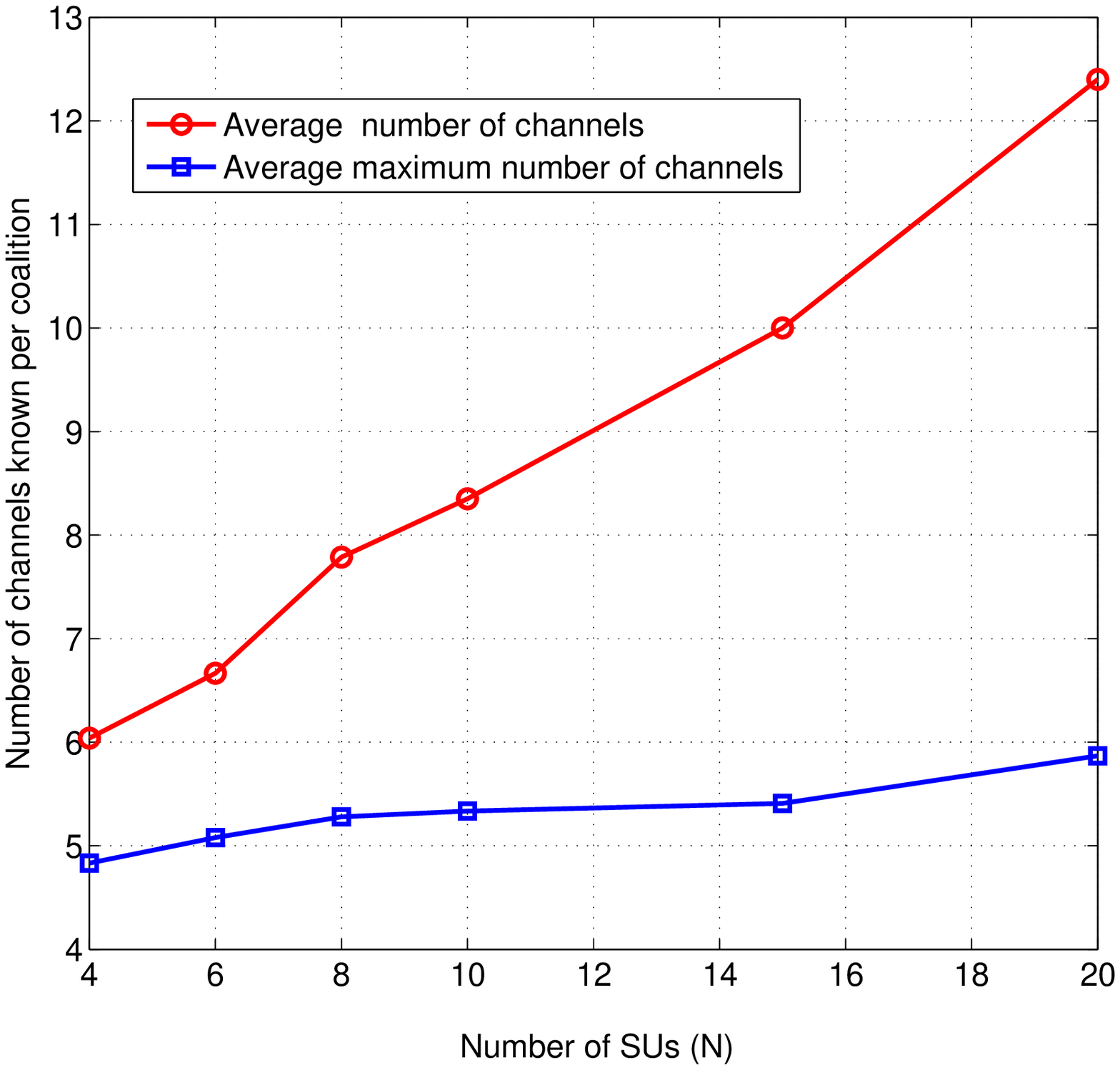}
\end{center}\vspace{-0.8cm}
\caption {Average and average maximum number of channels (averaged over random positions of the SUs and the random realizations of the probabilities $\theta_k, \forall k \in \mathcal{K}$) known per coalition for a network with $K=14$ channels as the network size $N$ varies.} \label{fig:numchan}
\end{figure}

In Fig.~\ref{fig:perfalpha}, we show the average payoff achieved per SU per slot for a network with $N=10$~SUs and $K=14$ channels as the fraction of time needed for sensing a single channel $\alpha$ increases. The results are averaged over random positions of the SUs and the random realizations of the probabilities $\theta_k, \forall k \in \mathcal{K}$. Fig.~\ref{fig:perfalpha} demonstrates that, as the amount of time $\alpha$ dedicated for sensing a single channel increases, the time that can be allotted for spectrum access is reduced, and, thus, the average payoff per SU per slot for both cooperative and non-cooperative spectrum sensing and access decreases. In this figure, we can see that, at all $\alpha$, the proposed joint spectrum sensing and access through coalition formation exhibits a performance gain over the non-cooperative case. This advantage decreases with $\alpha$, but it does not go below an improvement of $54.7\%$ relative to the non-cooperative scheme at $\alpha=0.5$, i.e., when half of the slot is used for sensing a single channel.

 Fig.~\ref{fig:perfchan} shows the average payoff achieved per SU per slot for a network with $N=10$ SUs as the number of PU channels, $K$, increases. The results are averaged over random positions of the SUs and the random realizations of the probabilities $\theta_k, \forall k \in \mathcal{K}$. In this figure, we can see that as the number of channels $K$ increases, the performance of both cooperative and non-cooperative spectrum sensing and access increases. For the non-cooperative case, this increase is mainly due to the fact that, as more channels become available, the possibility of interference due to the non-cooperative channel selection is reduced. For the proposed coalition formation algorithm, the increase in the performance is also due to the increased number of channels that the SUs can share as $K$ increases. Furthermore,  Fig.~\ref{fig:perfchan}  demonstrates that the proposed joint spectrum sensing and access
 presents a significant performance advantage over the non-cooperative case  which is at least $63.5\%$ for $K=20$ and increases for networks with smaller channels. The increase in the performance advantage highlights the ability of the SUs to reduce effectively the mutual interference through the proposed coalition formation algorithm.

\subsection{Coalition Size and Known Channels}
In Fig.~\ref{fig:size}, we show the average and average maximum coalition size (averaged over the random positions of the SUs and the random realizations of the probabilities $\theta_k, \forall k \in \mathcal{K}$) resulting from the proposed algorithm as the number of SUs, $N$, increases, for a network with $K=14$~channels. Fig.~\ref{fig:size} shows that, as $N$ increases, both the average and maximum coalition size increase with the average having a smaller slope. Further, we note that the average and average maximum coalition size reach around $3.2$ and $8$ at $N=20$, respectively. Hence, Fig.~\ref{fig:size} demonstrates that, although some large coalitions are emerging in the network, on the average, the size of the coalitions is relatively small. This result is due to the fact that, as mentioned in  Section~\ref{sec:jssa}, although cooperation is beneficial, it is also accompanied by costs due to the needed re-ordering of the channels, the occurrence of new interference due to channel sharing, and so on. These costs limit the coalition size on the average. Thus, Fig.~\ref{fig:size} shows that, when using coalition formation for joint spectrum sensing and access, the resulting network is, in general, composed of a large number of small coalitions with the occasional formation of large coalitions. In brief,  Fig.~\ref{fig:size} provides insight into the network structure when the SUs cooperate for joint spectrum sensing and access.

Fig.~\ref{fig:numchan} shows the average and average maximum number of channels known per coalition (averaged over the random positions of the SUs and the random realizations of the probabilities $\theta_k, \forall k \in \mathcal{K}$) as the number of SUs, $N$, increases, for a network with $K=14$~channels. Fig.~\ref{fig:numchan} demonstrates that both the average and average maximum number of known channels per coalition increase with the network size $N$. This increase is due to the fact that, as more SUs are present in the network, the cooperation possibilities increase and the number of channels that can be shared per coalition also increases. In this regard, the average number of known channels ranges from around $4.8$ for $N=4$ to around $5.9$ for $N=20$, while the average maximum goes from $6$ at $N=4$ to $12.4$ at $N=20$. This result shows that the increase in the average number of known channels is small while that of the maximum is more significant. This implies that, due to the cooperation tradeoffs, in general, the SUs have an incentive to share a relatively moderate number of channels with the emergence of few coalitions sharing a large number of channels.

\subsection{Adaptation to Environmental Changes}
In Fig.~\ref{fig:theta}, we show, over a period of $4$ minutes (after the initial network formation), the evolution of a network of $N=10$~SUs and $K=14$~channels over time when the PUs' traffic, i.e., the probabilities $\theta_k,\forall k\in \mathcal{K}$ vary, independently, every $1$~minute. As the channel occupancy probability varies, the structure of the network changes, with new coalitions forming and others breaking up due to switch operations occurring. The network starts with a non-cooperative structure made up of $10$
non-cooperative SUs. First, the SUs self-organize
in $3$ coalitions upon the occurrence of $8$~switch operations as per Fig.~\ref{fig:theta} (at time $0$). With time, the SUs can adapt the network's structure to the changes in the traffic of the PUs through adequate switch operations. For example, after $1$ minute has elapsed, the number of coalitions increase from $3$ to $4$ as the SUs perform $5$ switch operations. After a total of $18$~switch operations over the $4$ minutes, the final partition is made up of $4$ coalitions that evolved from the initial $3$ coalitions.

In Fig.~\ref{fig:speed}, we show the average total number of switch operations per minute (averaged over the random positions of the SUs and the random realizations of the probabilities $\theta_k, \forall k \in \mathcal{K}$) for various speeds of the SUs for networks with $K=14$~channels and for the cases of $N=10$~SUs and $N=15$~SUs. The SUs are moving using a random walk mobility model for a period of $2.5$ minutes with the direction changing every $\eta=30$~seconds. As the velocity increases,  the average frequency of switch operations increases for all network sizes due to the dynamic changes in the network structure incurred by more mobility. These switch operations result from that fact that, periodically, every $\eta=30$~seconds, the SUs are able to reengage in coalition formation through Algorithm~\ref{alg:coalform}, adapting the coalitional structure to the changes due to mobility. The average total number  of switch operations per minute also increases with the number of SUs as the possibility of finding new cooperation partners becomes higher for larger $N$. For example, while for the case of $N=10$~SUs the average frequency of switch operations varies from $4.8$ operations per minute at a speed of $18$~km/h to $15.2$ operations per minute at a speed of $72$~km/h, for the case of $N=15$~SUs, the increase is much steeper and varies from $6.4$ operations per minute at $18$~km/h to $26$ operations per minute at $72$~km/h.

The network's adaptation to mobility is further assessed in Fig.~\ref{fig:life} where we show, over a period of $2.5$ minutes, the average coalition lifespan (in seconds) achieved for various speeds of the SUs in a cognitive network with $K=14$~channels and different number of SUs. The mobility model is similar to the one used in Fig.~\ref{fig:speed} with $\eta=30$~seconds. We define the coalition lifespan as the time (in seconds) during which a coalition operates in the  network prior to accepting new SUs or breaking into smaller coalitions (due to switch operations). Fig.~\ref{fig:life} shows that, as the speed of the SUs increases, the average lifespan of a coalition decreases due to the fact that, as mobility becomes higher, the likelihood of forming new coalitions or splitting existing coalitions increases. For example, for $N=15$~SUs, the coalition lifespan drops from around $69.5$~seconds for a velocity of $18$~km/h to around $53.5$~seconds at $36$~km/h, and down to about $26.4$ seconds at $72$~km/h. Furthermore, Fig.~\ref{fig:life} shows that as more SUs are present in the network, the coalition lifespan decreases. For instance, for any given velocity, the lifespan of a coalition for a network with $N=10$~SUs is larger than that of a coalition in a network with $N=15$~SUs. The main reason behind the decrease in coalition lifespan with $N$ is that, for a given speed, as $N$ increases, the SUs are more able to find new partners to join with as they move. In a nutshell, Fig.~\ref{fig:life} provides an interesting assessment of the topology adaptation aspect of the proposed coalition formation algorithm through switch operations.

Finally,  we note that, in order to highlight solely the changes due to mobility, the fading amplitude was considered constant in Fig.~\ref{fig:speed} and Fig.~\ref{fig:life}. Similar results can be seen when the fading amplitude also changes.

\section{Conclusions}\label{sec:conc}
In this paper, we have introduced a novel model for cooperation in cognitive radio networks, which accounts for both the spectrum sensing and spectrum access aspects. We have modeled the problem as a coalitional game in partition form and we have derived an algorithm that allows the SUs to make distributed decisions for joining or leaving a coalition, depending on their achieved utilities which account for the average time to find a unoccupied channel (spectrum sensing) and the average achieved capacity (spectrum access). We have shown that, by using the proposed coalition formation algorithm, the SUs can self-organize into a Nash-stable network partition, and adapt this topology to environmental changes such as a change in the traffic of the PUs or slow mobility. Simulation results have shown that the proposed algorithm yields gains, in terms of average payoff per SU per slot, reaching up to $86.8\%$ relative to the non-cooperative case for a network with $20$~SUs.
\def\baselinestretch{0.93}
\bibliographystyle{IEEEtran}
\bibliography{references}
\begin{figure}[e]
\begin{center}
\includegraphics[width=100mm]{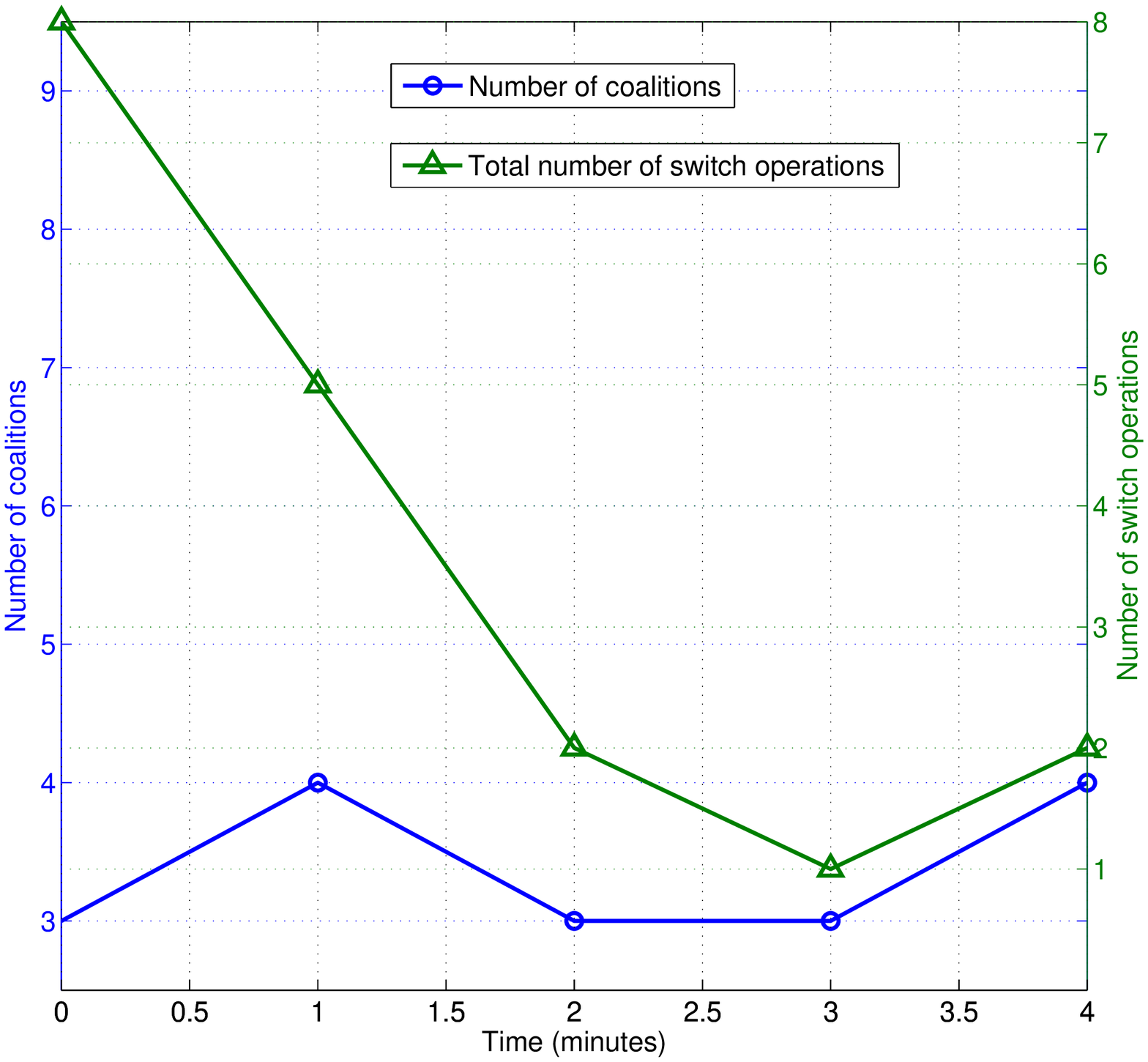}
\end{center}\vspace{-0.7cm}
\caption {Network structure evolution with time for $N=10$~SUs, as the traffic of the PUs, i.e., $\theta_k \forall k \in \mathcal{K}$ varies over a period of $4$~minutes.} \label{fig:theta}\vspace{-0.8cm}
\end{figure}

\begin{figure}[e]
\begin{center}
\includegraphics[width=100mm]{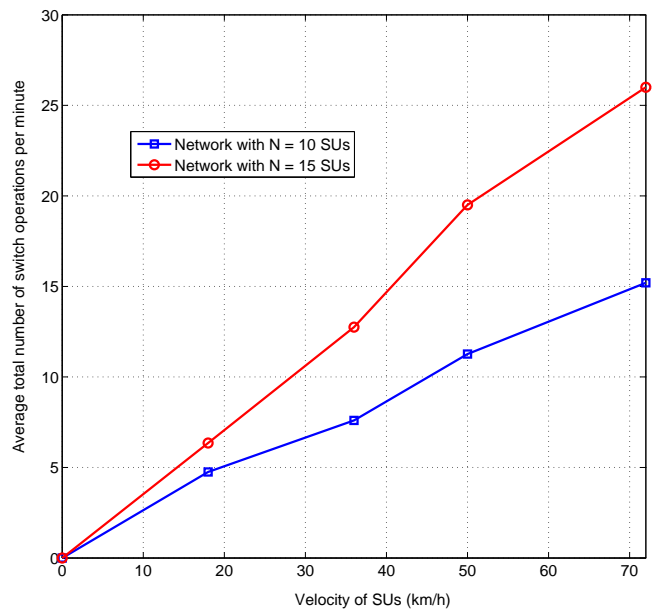}
\end{center}\vspace{-0.8cm}
\caption {Average frequency of switch operations per minute (averaged over random positions of the SUs and the random realizations of the probabilities $\theta_k, \forall k \in \mathcal{K}$)  for
different speeds in a network with $K=14$ channels for $N=10$~SUs and $N=15$~SUs.} \label{fig:speed}\vspace{-0.75cm}
\end{figure}
\begin{figure}[e]
\begin{center}
\includegraphics[width=100mm]{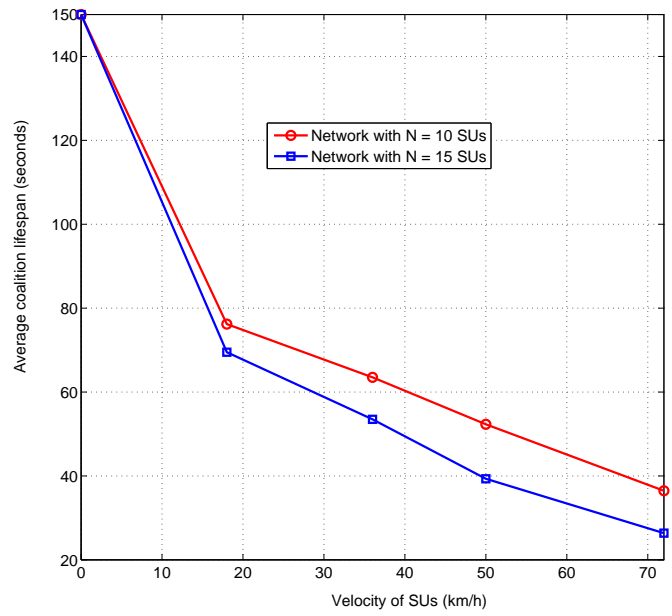}
\end{center}\vspace{-0.8cm}
\caption {Average coalition lifespan in seconds (averaged over random positions of the SUs and the random realizations of the probabilities $\theta_k, \forall k \in \mathcal{K}$) for
different speeds in a network with $K=14$ channels for $N=10$~SUs and $N=15$~SUs.} \label{fig:life}\vspace{-0.75cm}
\end{figure}

\end{document}